%% file: main.tex
\documentclass[11pt]{article}
\input{setting}

\title{Deterministic Monotone Min-Plus Product and Convolution}
\author{
Ce Jin\thanks{University of California Berkeley. \textbf{\texttt{cejin@berkeley.edu}}. Supported by the Miller Research Fellowship at the Miller Institute for Basic Research in Science, UC Berkeley.}
\and
Jaewoo Park\thanks{University of California San Diego. \textbf{\texttt{pjaewoo@ucsd.edu}}. Supported by NSF HDR TRIPODS Phase II grant 2217058 (EnCORE Institute).}
\and
Barna Saha\thanks{University of California San Diego. \textbf{\texttt{bsaha@ucsd.edu}}. Supported by NSF HDR TRIPODS Phase II grant 2217058 (EnCORE Institute).}
\and
Yinzhan Xu\thanks{University of California San Diego. \textbf{\texttt{xyzhan@ucsd.edu}}. Supported by NSF HDR TRIPODS Phase II grant 2217058 (EnCORE Institute).}
}

\date{\vspace{-5ex}}

\begin{document}
\pagenumbering{roman}
\maketitle

\begin{abstract}
\input{abstract}
\end{abstract}

\newpage
\pagenumbering{arabic}
\input{introduction}
\input{overview}
\input{preliminaries}
\input{product}

\input{column}
\input{convolution}

\newpage
\bibliographystyle{alphaurl} 
\bibliography{ref}
\appendix
\input{appendix}

\end{document}

%% file: setting.tex
\usepackage[margin=1in]{geometry}
\usepackage[T1]{fontenc}
\usepackage{times}
\usepackage{amsmath,amssymb,amsthm,mathtools}
\usepackage[shortlabels]{enumitem}
\usepackage{array}
\usepackage{tablefootnote}
\usepackage{xurl}
\usepackage[colorlinks,citecolor=blue,linkcolor=blue,urlcolor=red,pagebackref]{hyperref}
\usepackage[capitalise,nameinlink]{cleveref}
\crefname{enumi}{Difficulty}{Difficulties}
\theoremstyle{plain}
\newtheorem{prob}{Problem}
\newtheorem{theorem}{Theorem}[section]
\newtheorem{lemma}[theorem]{Lemma}
\newtheorem{corollary}[theorem]{Corollary}
\newtheorem{claim}[theorem]{Claim}
\newtheorem{fact}[theorem]{Fact}
\theoremstyle{definition}
\newtheorem{definition}[theorem]{Definition}

\newtheorem*{remark*}{Remark}

\newtheorem{innercustomprob}{Problem}
\newenvironment{customprob}[1]
{\renewcommand\theinnercustomprob{#1}\innercustomprob}{\endinnercustomprob}
\newcommand{\eps}{\varepsilon}
\renewcommand{\Pr}{\operatorname*{\mathbf{Pr}}}
\newcommand{\Ex}{\operatorname*{\mathbb{E}}}
\newcommand{\poly}{\operatorname{\mathrm{poly}}}
\newcommand{\polylog}{\poly\log}
\newcommand{\tcO}{\widetilde{\mathcal{O}}}
\newcommand{\hO}{\widehat{\mathcal{O}}}
\newcommand{\dia}{\diamond}
\newcommand{\lm}{{\ell_{\max}}}
\newcommand{\MM}{\operatorname*{MM}}
\newcommand{\high}{\mathrm{high}}
\newcommand{\low}{\mathrm{low}}

\newcommand{\cO}{\mathcal{O}}
\newcommand{\cP}{\mathcal{P}}

\newcommand{\cX}{\mathcal{X}}
\newcommand{\cY}{\mathcal{Y}}

\newcommand{\Z}{\mathbb{Z}}
\newcommand{\Ah}{A^{\high}}
\newcommand{\Bh}{B^{\high}}
\newcommand{\Ch}{C^{\high}}
\newcommand{\Al}{A^{\low}}
\newcommand{\Bl}{B^{\low}}
\newcommand{\Cl}{C^{\low}}

%% file: abstract.tex
The \emph{Monotone Min-Plus Product} problem is a useful primitive that has seen many algorithmic applications over the past decade. It also generalizes various other structured Min-Plus products studied in the literature, such as Bounded Difference Min-Plus Product and Bounded Integer Min-Plus Product. In this problem, we are given two $n\times n$ integer matrices $A$ and $B$, where each row of $B$ is a monotone non-decreasing sequence of integers from $\{1,\dots,n\}$, and the goal is to compute their Min-Plus product, defined as the $n\times n$ matrix $C$ with $C_{i,j} = \min_{k}\{A_{i,k} + B_{k,j}\}$. The fastest known algorithm for this task [Chi, Duan, Xie, and Zhang, STOC'22] runs in $n^{(\omega+3)/2+o(1)} = \cO(n^{2.686})$ time, significantly improving over the brute-force cubic algorithm. However, its main disadvantage is that it requires \emph{randomization}, which is then inherited by all downstream applications.

Our main result is a \emph{deterministic} algorithm for Monotone Min-Plus product with the same time complexity $n^{(\omega+3)/2+o(1)} = \cO(n^{2.686})$ as its randomized counterpart, improving upon the previous deterministic bound $\cO(n^{2.875})$ [Gu, Polak, Vassilevska Williams, and Xu, ICALP'21]. Our derandomization also applies to previously studied extensions and variants (e.g., [D\"urr, IPL'23]), including rectangular matrices, bounded range $[n^\mu]$, and column-monotone matrices. As an immediate consequence, we derandomize state-of-the-art algorithms for multiple problems, including Language Edit Distance, RNA Folding, Optimum Stack Generation, unweighted Tree Edit Distance, Batched Range Mode, and Approximate All-Pairs Shortest Paths.

Our techniques also yield a deterministic algorithm for the \emph{Monotone Min-Plus Convolution} problem that runs in $n^{1.5+o(1)}$ time, nearly matching the best-known randomized time complexity $\tcO(n^{1.5})$ [Chi, Duan, Xie, and Zhang, STOC'22]. This algorithm can be used to derandomize state-of-the-art algorithms for Jumbled Indexing for binary strings and several variants of Knapsack.

%% file: introduction.tex
\section{Introduction}
One of the most fundamental problems in algorithm design is the All-Pairs Shortest Paths problem (APSP), which asks us to compute pairwise distances in an $n$-node weighted graph. APSP is known to be asymptotically equivalent to \emph{Min-Plus product}~\cite{WilliamsW18jacm}, where we are given two $n \times n$ matrices $A$ and $B$ and need to compute an $n \times n$ matrix $A\star B$ such that
\[
    (A\star B)_{i,j} := \min_{1 \le k \le n} \{A_{i,k}+B_{k,j}\}.
\]
The fastest known algorithm for these problems runs in $n^3/2^{\Omega(\sqrt{\log n})}$ time, due to Williams~\cite{Williams18apsp}, and further improving this running time remains a major open problem. In fact, a popular hypothesis in fine-grained complexity is that APSP, or equivalently Min-Plus product, does not admit an $\cO(n^{3-\varepsilon})$-time algorithm for any $\varepsilon > 0$; see~\cite{virgisurvey} for a survey on fine-grained complexity.

Despite the difficulty of designing faster algorithms for Min-Plus product, researchers have identified several structured cases where it admits truly subcubic-time algorithms. For instance, \emph{bounded integer Min-Plus product}, where both matrices have entries from $[M]$\footnote{For a positive integer $m$, $[m]$ denotes $\{1,2,\ldots,m\}$.} for some integer $M$, can be computed in $\hO(Mn^{\omega})$ time~\cite{AlonGM97}.\footnote{In this paper, we use $\hO$ to hide $n^{o(1)}$ factors and $\tcO$ to hide $\polylog n$ factors.}\footnote{$\omega \le 2.372$~\cite{AlmanDWXXZ25} denotes the square matrix multiplication exponent; that is, $\omega$ is the smallest constant such that the product of two $n \times n$ matrices can be computed in $\hO(n^{\omega})$ arithmetic operations. More generally, we use $\omega(a,b,c)$ to denote the smallest constant such that the product of an $n^a \times n^b$ matrix and an $n^b \times n^c$ matrix can be computed in $\hO(n^{\omega(a,b,c)})$ arithmetic operations.} A generalization of bounded integer Min-Plus product is \emph{bounded difference Min-Plus product}, where adjacent entries in the input matrices differ by at most $M$. Bringmann, Grandoni, Saha, and Vassilevska Williams~\cite{DBLP:journals/siamcomp/BringmannGSW19} first studied this setting and gave an $\cO(n^{2.825})$-time randomized algorithm and an $\cO(n^{2.861})$-time deterministic algorithm for the case $M=\cO(1)$. They used bounded difference Min-Plus product as a tool to obtain the first subcubic-time algorithms for Language Edit Distance, RNA Folding, and Optimum Stack Generation.

An important setting that further generalizes bounded difference Min-Plus product, and is the focus of this paper, is \emph{monotone Min-Plus product}; see~\cite{DBLP:conf/icalp/Gu0WX21} for a discussion of this generalization. An $n \times n$ matrix $B$ is called row-monotone if
\begin{itemize}
    \item for every $1 \le i,j \le n$, we have $B_{i,j} \in [n]$;
    \item for every $1 \le i \le n$ and $1 \le j < n$, we have $B_{i,j} \le B_{i,j+1}$.
\end{itemize}
Column-monotone matrices are defined analogously. The task of computing the Min-Plus product between an arbitrary integer matrix $A$ and a row-monotone matrix $B$ is called the monotone Min-Plus product problem. The importance of monotone Min-Plus product is reflected in its wide range of applications, including Language Edit Distance~\cite{DBLP:journals/siamcomp/BringmannGSW19}, RNA Folding~\cite{DBLP:journals/siamcomp/BringmannGSW19}, Optimum Stack Generation~\cite{DBLP:journals/siamcomp/BringmannGSW19}, Batched Range Mode~\cite{DBLP:conf/soda/WilliamsX20, DBLP:conf/icalp/Gu0WX21, DBLP:journals/ipl/Durr23}, unweighted Tree Edit Distance~\cite{Mao21, DBLP:journals/ipl/Durr23, Nogler0SWX025}, Single-Source Replacement Paths~\cite{DBLP:conf/icalp/Gu0WX21, DBLP:journals/ipl/Durr23}, approximate APSP~\cite{DBLP:conf/icalp/DengKRWZ22, DBLP:journals/ipl/Durr23, DBLP:conf/soda/SahaY24}, and $k$-Dyck Edit Distance~\cite{DBLP:journals/talg/FriedGKKPS24, DBLP:journals/ipl/Durr23}.

The fastest algorithm for monotone Min-Plus product is due to Chi, Duan, Xie, and Zhang~\cite{CDXZ22} and runs in time $\hO(n^{(3+\omega)/2}) = \cO(n^{2.686})$. Unfortunately, this running time is achieved by a randomized algorithm, making the fastest algorithms for applications of monotone Min-Plus product randomized as well. More frustratingly, for many of these algorithms, the only source of randomness is the black-box call to monotone Min-Plus product; see \cref{tab:applications} and \cref{tab:2kapsp} for examples. Thus, prior to our work, obtaining deterministic algorithms for these applications required replacing the randomized monotone Min-Plus product algorithm with a slower deterministic one, such as the $\cO(n^{2.875})$-time algorithm of~\cite{DBLP:conf/icalp/Gu0WX21}, or the $\cO(n^{2.861})$-time bounded difference Min-Plus product algorithm of~\cite{DBLP:journals/siamcomp/BringmannGSW19}.

\begin{table}[t]
    \centering
    \begin{tabular}{lcc}
    \hline
     Problem & Randomized Runtime & Deterministic Runtime\\
    \hline
    \hline
       Language Edit Distance
       & $n^{(3+\omega)/2} \le n^{2.686}$~\cite{DBLP:journals/siamcomp/BringmannGSW19,CDXZ22}
       & $n^{2.861}$~\cite{DBLP:journals/siamcomp/BringmannGSW19}\\
        RNA Folding
        & $n^{(3+\omega)/2} \le n^{2.686}$~\cite{DBLP:journals/siamcomp/BringmannGSW19,CDXZ22}
        & $n^{2.861}$~\cite{DBLP:journals/siamcomp/BringmannGSW19}\\
        Optimum Stack Generation
        & $n^{(3+\omega)/2} \le n^{2.686}$~\cite{DBLP:journals/siamcomp/BringmannGSW19,CDXZ22}
        & $n^{2.861}$~\cite{DBLP:journals/siamcomp/BringmannGSW19}\\
        Batched Range Mode
        & $n^{(3+2\omega)/(3+\omega)} \le n^{1.458}$~\cite{DBLP:journals/ipl/Durr23}\tablefootnote{The running time for Batched Range Mode reported in the first version of this paper has been corrected. The exponent originally reported also appeared in \cite{DBLP:journals/ipl/Durr23} and was later revised in the full arXiv version.}
        & $n^{(18+2\omega)/(13+\omega)} \le n^{1.480}$~\cite{Gao022}\\
        Unweighted Tree Edit Distance
        & $n^{(3+\omega)/2} \le n^{2.686}$~\cite{Nogler0SWX025}
        & $n^{2.964}$~\cite{Mao21}\\
        $k$-Dyck Edit Distance
        & $n+k^{4.545}$~\cite{DBLP:journals/talg/FriedGKKPS24,DBLP:journals/ipl/Durr23}
        & $n+k^{4.854}$~\cite{DBLP:journals/talg/FriedGKKPS24}\\
        \hline
    \end{tabular}
    \caption{Applications where the monotone Min-Plus product oracle is the only source of randomization in the fastest known randomized algorithms. Previous best deterministic running times are also listed. The $\hO$-notation is omitted.}
    \label{tab:applications}
\end{table}

\begin{table}[t]
    \centering
    \begin{tabular}{cc}
    \hline
     Additive error $+2k$ & Runtime for $+2k$-approximate APSP \\
    \hline
    \hline
       $+2$  & $n^{2.22548}$  (randomized)~\cite{additiveapsp2026}  \\
       $+4$  & $n^{2.14613}$  (deterministic)~\cite{additiveapsp2026} \\
       $+6$  & $n^{2.10260}$  (deterministic)~\cite{additiveapsp2026} \\
       $+8$  & $n^{2.07727}$  (randomized)~\cite{DBLP:conf/soda/SahaY24} \\
       $+10$ & $n^{2.06194}$  (randomized)~\cite{DBLP:conf/soda/SahaY24} \\
       $\cdots$ & $\cdots$ (randomized)~\cite{DBLP:conf/soda/SahaY24}\\
       \hline
    \end{tabular}
     \caption{Fastest known algorithms for $+2k$-approximate APSP on unweighted undirected graphs. The algorithms of~\cite{DBLP:conf/soda/SahaY24} use the monotone Min-Plus product oracle as their only randomized component. The $\hO$-notation is omitted.}
     \label{tab:2kapsp}
\end{table}

\begin{remark*}
We briefly explain how some of the randomized running times in \cref{tab:applications} and \cref{tab:2kapsp} are obtained. The randomized running times for Language Edit Distance, RNA Folding, and Optimum Stack Generation in \cref{tab:applications} are obtained by replacing the bounded difference Min-Plus product algorithm of~\cite{DBLP:journals/siamcomp/BringmannGSW19} with the monotone Min-Plus product algorithm of~\cite{CDXZ22}. For every positive integer $k$, the algorithm of~\cite{DBLP:conf/soda/SahaY24} for $+2k$-approximate APSP in \cref{tab:2kapsp} runs in $\hO(n^{2+x/(k+1)})$ time, where $x$ is the solution to $1+2x = \omega(1-\frac{k-1}{k+1}x,1-x,1-\frac{k-2}{k+1}x)$, and uses the monotone Min-Plus product oracle as its only source of randomness. For $2k \ge 8$, this is the fastest known algorithm. The work of~\cite{additiveapsp2026} improves the running times of the $+4$- and $+6$-approximation algorithms of~\cite{DBLP:conf/soda/SahaY24} using deterministic algorithms, and further improves the running time for $+2$-approximation using a randomized algorithm whose randomness comes from a source other than the monotone Min-Plus product oracle.
\end{remark*}

A problem closely related to Min-Plus product is Min-Plus convolution. In the \emph{Min-Plus convolution} problem, the input consists of two length-$n$ arrays $A=(A_1,A_2,\ldots,A_n)$ and $B=(B_1,B_2,\ldots,B_n)$, and the output is an array $A\diamond B$ indexed by $k=2,\dots,2n$, where
\[
    (A\diamond B)_k=\min_{1\le i\le k-1}\{A_i+B_{k-i}\},
\]
with out-of-bounds entries set to $\infty$. The fastest algorithm for this problem runs in $n^2/2^{\Theta(\sqrt{\log n})}$ time, which can be obtained by combining Williams' APSP algorithm~\cite{Williams18apsp} with the reduction in~\cite{BremnerCDEHILPT14}. A popular hypothesis in fine-grained complexity is the Min-Plus convolution hypothesis, which postulates that Min-Plus convolution cannot be solved in $\cO(n^{2-\varepsilon})$ time for any $\varepsilon>0$; see~\cite{CyganMWW17,KunnemannPS17}. It is known that the Min-Plus convolution hypothesis implies the Min-Plus product hypothesis~\cite{BremnerCDEHILPT14}.

Similar to Min-Plus product, Min-Plus convolution admits significantly faster algorithms when the input is structured. A length-$n$ array $A$ is called monotone if
\begin{itemize}
    \item for every $1 \le i \le n$, we have $A_i \in [n]$;
    \item for every $1 \le i < n$, we have $A_i \le A_{i+1}$.
\end{itemize}
When both input arrays are monotone, Chan and Lewenstein~\cite{DBLP:conf/stoc/ChanL15} gave an $\cO(n^{1.859})$-time randomized algorithm and an $\cO(n^{1.864})$-time deterministic algorithm. Chi, Duan, Xie, and Zhang~\cite{CDXZ22} improved the randomized running time to $\tcO(n^{1.5})$. Bringmann, D{\"u}rr, and Polak~\cite{DBLP:conf/esa/BringmannD024} generalized this result to a randomized $\tcO(n^{1+\mu/2})$-time algorithm for arrays with entries in $[n^\mu]$ and with only one of the input arrays being monotone.

Monotone Min-Plus convolution also has a wide range of applications. Examples include Jumbled Indexing~\cite{DBLP:conf/stoc/ChanL15}, Knapsack problems with bounded weights and profits~\cite{BringmannC22icalp,DBLP:conf/esa/BringmannD024}, and weak approximation schemes for Knapsack problems~\cite{BringmannC22icalp,chenweakapproxknap}.\footnote{These Knapsack results use Max-Plus convolution instead of Min-Plus convolution, but the two are equivalent up to a sign change.} Due to the gap between the best-known randomized and deterministic algorithms for monotone Min-Plus convolution, the best-known running times for all these applications are achieved by randomized algorithms. For some applications, the monotone Min-Plus convolution subroutine is the only remaining source of randomness; see \cref{tab:applications-convo}.

\begin{table}[t]
    \centering
    \setlength{\tabcolsep}{3pt}
    \begin{tabular}{@{}>{\raggedright\arraybackslash}m{17.8em}cc@{}}
    \hline
     Problem & Randomized Runtime & Deterministic Runtime\\
    \hline
    \hline
       Jumbled Indexing for Binary Strings
       & $n^{1.5}$~\cite{DBLP:conf/stoc/ChanL15, CDXZ22}
       & $n^{1.864}$~\cite{DBLP:conf/stoc/ChanL15} \\
        Unbounded Knapsack
        & $n + w_{\max}\sqrt{p_{\max}}$~\cite{DBLP:conf/esa/BringmannD024}
        & $n + (p_{\max}+w_{\max})^{1.864}$~\cite{BringmannC22icalp} \\
        0-1 Knapsack 
        & $n + W\sqrt{OPT}$~\cite{DBLP:conf/esa/BringmannD024}
        & $n + W\cdot OPT^{0.937}$~\cite{BringmannC22icalp,chan26} \\
        Weakly-Approximate Unbounded Knapsack
        & $n + 1/\eps^{1.5}$~\cite{BringmannC22icalp}
        & $n + 1/\eps^{1.864}$~\cite{BringmannC22icalp} \\
    \end{tabular}
    \caption{Applications of monotone Min-Plus convolution for which, after applying known derandomizations where needed, the only remaining source of randomization is the monotone Min-Plus convolution oracle. The table also lists the previously best-known deterministic running times. For Knapsack problems, $w_{\max}$ and $p_{\max}$ bound item weights and profits, respectively; $W$ is the capacity and $OPT$ is the optimal profit. The $\tcO$-notation is omitted.}
    \label{tab:applications-convo}
\end{table}

\begin{remark*}
We explain how some of the running times in \cref{tab:applications-convo} are obtained. The randomized running time for Jumbled Indexing for Binary Strings follows by replacing the monotone Min-Plus convolution algorithm in~\cite{DBLP:conf/stoc/ChanL15} with the faster randomized algorithm of~\cite{CDXZ22}. For Unbounded Knapsack, the randomized $\tcO(n+(w_{\max}+p_{\max})^{1.5})$-time algorithm of Bringmann and Cassis~\cite{BringmannC22icalp}, combined with the rectangular and $[n^\mu]$-bounded generalization of~\cite{DBLP:conf/esa/BringmannD024}, yields the refined bounds $\tcO(n+w_{\max}\sqrt{p_{\max}})$ and $\tcO(n+p_{\max}\sqrt{w_{\max}})$.

For $0$-$1$ Knapsack, the randomized $\tcO(n+W\sqrt{OPT})$ bound follows from~\cite[Theorem~12]{DBLP:conf/esa/BringmannD024}. We state this bound, rather than their stronger $\tcO(n+W\sqrt{p_{\max}})$ bound, because it is the one to which our derandomization applies. The deterministic 0-1 Knapsack bound shown in the table is not explicitly stated in the literature. The $0$-$1$ Knapsack algorithm of~\cite{BringmannC22icalp} has another source of randomness, namely the random partitioning technique originating from~\cite{Bringmann17,CyganMWW17}. This additional source of randomness was recently derandomized by~\cite{chan26}; see~\cite[Appendix~A]{chan26} for a deterministic reduction from $0$-$1$ Knapsack to Min-Plus convolution in~\cite{CyganMWW17}. The same derandomization also applies to the random partitioning step used in~\cite{BringmannC22icalp}. More specifically, if monotone Min-Plus convolution with $[n^\mu]$-bounded inputs can be solved deterministically in $\tcO(n^{1+\alpha\mu})$ time, then the results of~\cite{chan26,BringmannC22icalp} imply a deterministic $0$-$1$ Knapsack algorithm running in $\tcO(n+W\cdot OPT^\alpha)$ time. The previously known deterministic algorithm for $[n^\mu]$-bounded monotone Min-Plus convolution can be obtained using the deterministic $\cO(n^{1.864})$-time algorithm for the $[n]$-bounded case~\cite{DBLP:conf/stoc/ChanL15} as a black box, together with blocking and balancing (details omitted). This yields a deterministic $\tcO(n^{1 + 2\mu/(4 - 1.864)})=\tcO(n^{1+0.937\mu})$-time algorithm for $[n^\mu]$-bounded monotone Min-Plus convolution, and therefore yields the deterministic $\tcO(n+W\cdot OPT^{0.937})$ bound shown in the table.
\end{remark*}

\subsection{Our Results}
We close the gap between the deterministic and randomized running times for monotone Min-Plus product and monotone Min-Plus convolution. Our work is inspired by the randomized algorithm of Chi, Duan, Xie, and Zhang~\cite{CDXZ22}. However, our presentation is more streamlined due to several reductions. Moreover, our algorithm for monotone Min-Plus product applies to rectangular matrices and to the case where the entries of $B$ lie in an arbitrary range $[n^\mu]$, rather than only in $[n]$. This matches the extensions studied in the randomized setting by~\cite{DBLP:journals/ipl/Durr23,DBLP:journals/talg/FriedGKKPS24,DBLP:conf/soda/SahaY24}, which generalized the framework of~\cite{CDXZ22}.

\begin{theorem}
\label{thm:main}
There is a deterministic algorithm that, given an $n^a\times n^b$ integer matrix $A$ and an $n^b\times n^c$ row-monotone matrix $B$ with entries in $[n^\mu]$, computes their Min-Plus product in time 
\[
    \hO(n^{a+c} + n^{b + c} + n^{(a+b+\mu+\omega(a,b,c))/2}).
\]
\end{theorem}
Note that the terms $n^{a+c}$ and $n^{b+c}$ in the running time above are necessary for reading the input and writing the output. As in~\cite{CDXZ22}, our algorithm also extends to the case where $B$ is column-monotone. Owing to our more streamlined algorithm, this extension is more modular than in prior approaches. The running time differs from the row-monotone case, consistent with the literature~\cite{DBLP:journals/ipl/Durr23}.

\begin{theorem}
\label{thm:main-column}
There is a deterministic algorithm that, given an $n^a\times n^b$ integer matrix $A$ and an $n^b\times n^c$ column-monotone matrix $B$ with entries in $[n^\mu]$, computes their Min-Plus product in time
\[
    \hO\left(n^{a+b}+n^{b+c}+n^{(a+c+\mu+\omega(a,b,c))/2}\right).
\]
\end{theorem}

Very recently, it was shown in~\cite{Fischer26} that, unless the APSP hypothesis fails, there is no $\cO(n^{2.5-\eps})$-time algorithm for column-monotone Min-Plus product for any $\eps > 0$ when $a=b=c=\mu=1$. Thus, if $\omega=2$, our algorithm is near-optimal in this case under the APSP hypothesis.

As in~\cite{CDXZ22}, our techniques also extend to monotone Min-Plus convolution. Our algorithm for monotone Min-Plus convolution applies when the array entries lie in $[n^\mu]$. Its running time essentially matches the best-known randomized running time~\cite{CDXZ22,DBLP:conf/esa/BringmannD024}.

\begin{theorem}
\label{thm:main-conv}
There is a deterministic algorithm that, given two length-$n$ monotone arrays $A$ and $B$, both with entries in $[n^\mu]$, computes their Min-Plus convolution in time 
\[
\hO\left(n^{1+\mu/2}\right).
\]
\end{theorem}

Using the reduction of~\cite[Theorem~8]{DBLP:conf/esa/BringmannD024} from the case where both arrays are monotone to the case where only one array is monotone, our algorithm also applies when only one of the two arrays is monotone.

Our algorithms can be directly applied to many problems, derandomizing the best randomized algorithms.

\begin{corollary}
All problems in \cref{tab:applications,tab:applications-convo}, as well as $+2k$-approximate APSP for $2k\ge 8$ in \cref{tab:2kapsp}, can be solved by deterministic algorithms with the listed randomized running times, up to an $n^{o(1)}$ factor.
\end{corollary}

In fact, we can also derandomize the $+2$-approximate APSP algorithm of~\cite{additiveapsp2026} shown in \cref{tab:2kapsp}. This algorithm does not use monotone Min-Plus product, so its derandomization is not an immediate consequence of our algorithms; nevertheless, it can be derandomized using standard techniques. The proof can be found in \cref{appendix:apsp}. 

%% file: overview.tex
\section{Technical Overview}
\label{sec:overview}
In this technical overview, we focus on row-monotone Min-Plus product in the parameter regime originally considered by~\cite{CDXZ22}, where both input matrices have dimension $n\times n$ and entries from $[n]$. The extensions to rectangular matrices, larger entry ranges, column-monotone matrices, and monotone Min-Plus convolution are handled in the main body.

Our algorithm is a derandomization of the algorithm of~\cite{CDXZ22}. The only source of randomization in their algorithm is to pick a random prime $Q$ from some appropriate range. Hence, it may appear straightforward to derandomize their algorithm because there is now a standard recipe to derandomize the choice of a prime (see, e.g.,~\cite{DBLP:conf/stoc/ChanL15,FischerK024}): Instead of choosing a prime from a relatively large range, one takes $Q$ to be the product of several smaller primes from a smaller range, which are found one by one. For each such prime, one can afford to deterministically enumerate all primes in the smaller range and choose one that works. However, this approach faces several difficulties when applied to the algorithm of~\cite{CDXZ22}.

\begin{enumerate}[A]
    \item 
    \label{item:first-difficulty}
    The intuition for why \cite{CDXZ22} needs a random prime $Q$ is the following. There is an implicit set $S$ of roughly $n^3$ nonzero integers, and we want only a small number of integers in $S$ to be divisible by $Q$; call this quantity $f_Q$. When $Q$ is a random prime from $[M]$, the expected value of $f_Q$ can be upper bounded by roughly $n^3 / M$. If we aim to use the standard approach to derandomize the choice of $Q$, we would keep finding primes $p_i$ from a smaller range $[R]$, and maintain $Q = \prod_{i=1}^t p_i$, where $t$ is the current number of primes we have determined. For each additional prime, we expect $f_Q$ to drop by a factor of roughly $R$. In theory, one can show that such a prime exists, and even a random prime likely works. However, in order for a deterministic algorithm to pick $p_t$, it would naively need to compute the value of $f_Q$ for each choice of $p_t$. Unfortunately, there does not seem to be an efficient algorithm for this task, which marks the first difficulty in applying this derandomization approach. 
    \item 
    \label{item:second-difficulty}
    In addition to the previous role, the prime $Q$ in \cite{CDXZ22} also serves another role: it is used to separate an integer $x$ into a high-order part ($\lfloor x/Q\rfloor$) and a low-order part ($x \bmod Q$). If we use the standard approach to derandomize the choice of $Q$ by multiplying small primes one by one, the high-order parts and low-order parts of the integers keep changing, making it difficult to derandomize the algorithm. 
\end{enumerate}

\paragraph*{Standard Reductions.} To address \cref{item:second-difficulty}, we separate the two roles played by $Q$. As it turns out, it suffices to pick an arbitrary integer from an appropriate range to separate the integers into high-order parts and low-order parts. Let $M$ be such an integer with value $\Theta(n^{(3-\omega)/2})$. By standard reductions, it suffices to solve the following problem.
\begin{prob}
\label{prob:exact-tri-overview}
    Let $A, B, C$ be $n \times n$ integer matrices with nonnegative entries bounded by $\cO(n)$, where both $B$ and $C$ are row-monotone. In addition, all entries $x$ of the matrices satisfy $x \bmod{M} \le M / 10$. For every $i, j \in [n]$, determine whether there exists $k \in [n]$ such that $A_{i, k} + B_{k, j} = C_{i, j}$. 
\end{prob}
At a high level, the reduction has two steps. 
First, by recursively halving the entries of $A$ and $B$, the product of the halved instance gives an approximation to each entry of $A\star B$: if $C'$ is the product of the halved instance, then each true value $(A\star B)_{i,j}$ must be one of $2C'_{i,j}$, $2C'_{i,j}+1$, or $2C'_{i,j}+2$. 
Thus, a verifier for candidate values can recover the full product by testing these three candidates. 
Second, using a standard residue-shifting reduction similar to~\cite{CDXZ22}, we transform the verification instance into constantly many promised instances in which the relevant residues of $A$, $B$, and $C$ lie in a small interval modulo $M$. 

These standard reductions were also present in the algorithm of~\cite{CDXZ22}, appearing as a recursive step. By explicitly applying the standard reductions first, we move the recursive step out of the main algorithm, making the algorithm more streamlined. In addition, this helps unify the algorithms for the row-monotone and column-monotone cases. As it turns out, the column-monotone case can also be reduced to a version of \cref{prob:exact-tri-overview}, where for every $i, k$, we need to determine the existence of such a $j$; see \cref{sec:column} for details.

\paragraph*{Main algorithm.}
Define $A^{\high}_{i,k}=\lfloor A_{i,k}/M\rfloor$, and define $B^{\high}$ and $C^{\high}$ analogously. 
The algorithm first finds a modulus $Q$ satisfying $M \le Q \le M n^{o(1)}$; we defer the additional properties required of $Q$, as well as the description of how to find it, to later in the overview. 
Then the algorithm computes the following quantities:
\begin{itemize}
    \item $s_{i,j}$, the number of $k \in [n]$ such that $A_{i,k}+B_{k,j}\equiv C_{i,j} \pmod Q$;
    \item $s'_{i,j}$, the number of $k \in [n]$ such that $A_{i,k}+B_{k,j}\equiv C_{i,j} \pmod Q$ and $A^\high_{i,k}+B^\high_{k,j}\ne C^\high_{i,j}$.
\end{itemize}
It is not difficult to see that $A_{i,k}+B_{k,j}=C_{i,j}$ if and only if $A_{i,k}+B_{k,j}\equiv C_{i,j}\pmod Q$ and $A^\high_{i,k}+B^\high_{k,j}=C^\high_{i,j}$. Hence, for each pair $(i,j)$, there exists a witness $k$ if and only if $s_{i,j}>s'_{i,j}$. Therefore, it suffices to compute $s_{i,j}$ and $s'_{i,j}$ for every pair $(i,j)$. 

Using an approach similar to~\cite{AlonGM97}, one can compute all values $s_{i,j}$ in $\hO(Qn^{\omega})=\hO(n^{(3+\omega)/2})$ time. To discuss how to compute $s'_{i,j}$, we define segments, similarly to~\cite{CDXZ22}, and a special type of segment called active segments. Let $\lm := \lceil \log(M/20) \rceil$.

\begin{definition} 
For $0 \le \ell \le \lm$, a level-$\ell$ segment is a tuple $(i,k,[j_0,j_1])$ such that, for every $j_0 \le j \le j_1$, we have $\lfloor B_{k, j_0} / 2^\ell\rfloor = \lfloor B_{k, j} / 2^\ell\rfloor$ and $\lfloor C_{i, j_0} / 2^\ell\rfloor = \lfloor C_{i, j} / 2^\ell\rfloor$, and $[j_0,j_1]$ cannot be further extended.
\end{definition}

\begin{definition} 
For $0 \le \ell \le \lm$, a level-$\ell$ segment $(i,k,[j_0,j_1])$ is called active with respect to $Q$ if $A_{i, k}^\high + B_{k,j_0}^\high - C_{i, j_0}^\high \ne 0$ and $A_{i, k} + B_{k,j_0} - C_{i, j_0} \bmod{Q} \in [-4 \cdot 2^\ell, 4\cdot 2^\ell]$. Let $S_{\ell}(Q)$ be the set of active level-$\ell$ segments.
\end{definition}
Since $B$ and $C$ are row-monotone, the number of level-$\ell$ segments is $\cO(n^3 / 2^\ell)$.

We will find a modulus $Q$ so that, for every $0 \le \ell \le \lm$, $|S_{\ell}(Q)| \le \hO(n^{3} / Q)$. For now, assume that we have such a modulus $Q$. We show how to compute $s'_{i,j}$ under this assumption.

To see why active segments help compute $s'_{i,j}$, consider a triple $(i,j,k)$ such that $A_{i,k}+B_{k,j}\equiv C_{i,j}\pmod Q$ and $A^\high_{i,k}+B^\high_{k,j}\ne C^\high_{i,j}$; that is, $k$ contributes one to $s'_{i,j}$. Let $(i,k,[j_0,j_1])$ be the unique level-$0$ segment where $j \in [j_0, j_1]$. It is not difficult to see that this segment is active. Hence, we proceed by computing $S_0(Q)$:
\begin{itemize}
    \item First, we find $S_{\lm}(Q)$ by enumerating all level-$\lm$ segments and checking whether each segment is active. The number of level-$\lm$ segments is $\cO(n^3 / 2^\lm) = \cO(n^{(3+\omega)/2})$. Hence, this part takes $\tcO(n^{(3+\omega)/2})$ time.
    \item Suppose we have computed the set $S_{\ell+1}(Q)$ of active level-$(\ell+1)$ segments. For every segment $L\in S_{\ell+1}(Q)$, we enumerate all level-$\ell$ segments contained in $L$ and check whether they are active. Here, a segment $(i,k,[j_0,j_1])$ is contained in another segment $(i',k',[j_0',j_1'])$ if $i=i'$, $k=k'$, and $[j_0,j_1]\subseteq [j_0',j_1']$. It can be shown that every segment in $S_\ell(Q)$ is contained in some segment in $S_{\ell+1}(Q)$, so this procedure correctly computes $S_\ell(Q)$. Moreover, each level-$(\ell+1)$ segment contains only a constant number of level-$\ell$ segments. By the assumption $|S_{\ell}(Q)| \le \hO(n^3/Q)$ for every $\ell$, this step takes $\hO(n^3/Q)=\hO(n^{(3+\omega)/2})$ time. 
\end{itemize}
Recall that the segments in $S_0(Q)$ contain all triples $(i,j,k)$ that can contribute to $s'_{i,j}$. For every level-$0$ segment $(i,k,[j_0,j_1])\in S_0(Q)$, the values $B_{k,j}$ are the same for all $j\in[j_0,j_1]$, and the values $C_{i,j}$ are also the same for all $j\in[j_0,j_1]$. Hence, either all triples $(i,k,j)$ with $j\in[j_0,j_1]$ contribute to $s'_{i,j}$, or none of them do. We can therefore use a simple data structure to compute all values $s'_{i,j}$ from $S_0(Q)$ in $\tcO(|S_0(Q)|)=\hO(n^{(3+\omega)/2})$ time.

\paragraph*{Finding a Good Modulus.} It remains to find a modulus $Q \in [M, M n^{o(1)}]$ such that $|S_{\ell}(Q)| \le \hO(n^{3} / Q)$
for every $0 \le \ell \le \lm$. To do so, consider the following multiset of numbers for every level $\ell$:
\[
    \cX_{\ell} := \{A_{i, k} + B_{k, j_0} - C_{i, j_0} - s: (i, k, [j_0, j_1]) \text{ is a level-}\ell \text{ segment}, s \in [-4 \cdot 2^\ell, 4 \cdot 2^\ell] \} \cap (\Z \setminus \{0\}), 
\]
and let $X_{\ell, Q}$ denote the number of elements in $\cX_{\ell}$ that are divisible by $Q$. Clearly, $|\cX_{\ell} | \le \cO(n^3/2^\ell)\cdot \cO(2^\ell) = \cO(n^3)$. When $Q \ge M$ (and hence $Q > 4\cdot 2^\ell$), $X_{\ell, Q}$ is exactly the number of level-$\ell$ segments $(i, k, [j_0, j_1])$ such that $A_{i, k} + B_{k, j_0} - C_{i, j_0} \bmod{Q} \in [-4 \cdot 2^\ell, 4 \cdot 2^\ell]$ and $A_{i, k} + B_{k, j_0} - C_{i, j_0} \notin [-4 \cdot 2^\ell, 4 \cdot 2^\ell]$. Using the assumptions in \cref{prob:exact-tri-overview}, it is not difficult to show that one can equivalently replace the second condition with $A^\high_{i, k} + B^\high_{k, j_0} \ne C^{\high}_{i, j_0}$. Hence, $X_{\ell, Q}$ is exactly $|S_{\ell}(Q)|$.

To summarize, we have several multisets $\cX_{\ell}$, each containing $\cO(n^3)$ nonzero integers, and we need to find a modulus $Q$ such that, for every $\ell$, the number of elements of $\cX_{\ell}$ divisible by $Q$ is at most $\hO(n^3/Q)$. This is exactly the setting where the standard recipe of constructing $Q$ as a product of smaller primes applies. More precisely, we construct $Q$ as the product of primes $p_1,p_2,\ldots,p_t$ from $[R/2,R]$, for some relatively small $R=n^{o(1)}$. Let $P$ be the set of primes in $[R/2,R]$. At each step, we find an additional prime $p_i\in P$ by enumerating all primes in $P$ and choosing a suitable one. We stop once the product of the chosen primes exceeds $M$. By a standard probability analysis, there exist primes $p_1,p_2,\ldots,p_t\in P$ such that if $Q_t = \prod_{i=1}^t p_i$, then $X_{\ell,Q_t}$ is bounded by roughly $n^3/R^t$ for every $\ell$. Since $R^t\approx Q_t$, this gives $X_{\ell,Q_t}\lesssim n^3/Q_t$ for every $\ell$, as desired. Hence, it remains to find these suitable primes deterministically.

In this overview, we illustrate our approach by showing how to find the first prime $p\in P$ such that $X_{\ell,p}=\tcO(n^3/R)$ for every $\ell$. The same idea extends easily to finding the remaining primes. Since $R=n^{o(1)}$, we can afford to enumerate all primes in $P$ and choose one that works. However, even when $p$ is fixed, it is unclear how to compute $X_{\ell,p}$ efficiently: the straightforward approach requires enumerating all level-$\ell$ segments, which takes $\cO(n^3 / 2^\ell)$ time. When $\ell$ is small, this bound can be as high as $\cO(n^3)$. This reflects the first difficulty mentioned in \cref{item:first-difficulty}, although the setting here is already somewhat different.

Instead of computing $X_{\ell, p}$, we take a more indirect approach. Let
\[
    \cY_{\ell} := \{A_{i, k} + B_{k, j_0} - C_{i, j_0} - s: (i, k, [j_0, j_1]) \text{ is a level-}\ell \text{ segment}, s \in [-4 \cdot 2^\ell, 4 \cdot 2^\ell] \}.
\]
This multiset has the same definition as $\cX_{\ell}$, except that we do not remove the zero elements. Let $Y_{\ell,p}$ denote the number of elements in $\cY_{\ell}$ that are divisible by $p$. As it turns out, one can compute $Y_{\ell,p}$ efficiently in $\hO(2^\ell n^\omega)=\hO(Mn^\omega)$ time, following the high-level idea of~\cite{AlonGM97}. Next, observe that $Y_{\ell,p}-X_{\ell,p}$ is exactly the number of pairs $((i,k,[j_0,j_1]),s)$ such that $A_{i,k}+B_{k,j_0}-C_{i,j_0}-s=0$, where $(i,k,[j_0,j_1])$ is a level-$\ell$ segment and $s \in [-4\cdot 2^\ell,4\cdot 2^\ell]$. In particular, $Y_{\ell,p}-X_{\ell,p}$ does not depend on $p$; denote this quantity by $Z_\ell$.

In our algorithm, we compute $Y_\ell^* := \min_{p \in P} Y_{\ell,p}$ for every $\ell$, and then choose a prime $p_0 \in P$ minimizing $\max_\ell \{Y_{\ell,p_0}-Y_\ell^*\}$. This prime $p_0$ can be found deterministically in $M n^{\omega+o(1)}$ time, since this procedure only uses the computable values $Y_{\ell,p}$, rather than the values $X_{\ell,p}$. Note that $Y_\ell^*-Z_\ell$ is exactly $X_\ell^* := \min_{p\in P} X_{\ell,p}$.  Hence, $p_0$ minimizes 
\[
    \max_{0 \le \ell \le \lm} \{Y_{\ell, p} - Y_\ell^*\} = \max_{0 \le \ell \le \lm} \{X_{\ell, p} - X_\ell^*\}.
\]
By a standard probabilistic argument, $X_\ell^*\le \tcO(n^3/R)$ for every $\ell$. Thus, $p_0$ minimizes $\max_{0\le \ell \le \lm}\{X_{\ell,p}\}$ up to a $\tcO(n^3 / R)$ additive error. Again by a standard probabilistic argument,
\[
    \min_{p\in P}\max_{0\le \ell \le \lm}\{X_{\ell,p}\}\le \tcO(n^3/R).
\]
Therefore, $X_{\ell,p_0}\le \tcO(n^3/R)$ for every $\ell$, as desired.

%% file: preliminaries.tex
\section{Preliminaries}

\paragraph*{Notations.}
For integers $a$ and $M \ge 1$, we write $a \bmod M$ for the unique integer $b \in \{0,1,\ldots,M-1\}$ such that $a \equiv b \pmod M$. For a positive integer $x$, we write $[x] := \{1,2,\ldots,x\}$. For integers $x \le y$, we write $[x,y] := \{x,x+1,\ldots,y\}$.

\paragraph*{Min-Plus Product.}
Let $A$ be an $n^a \times n^b$ integer matrix and let $B$ be an
$n^b \times n^c$ integer matrix. The \emph{Min-Plus product} of $A$ and $B$ is the $n^a \times n^c$ matrix
$A \star B$ defined by
\[ 
    (A \star B)_{i,j} := \min_{1 \le k \le n^b} \{ A_{i,k} + B_{k,j} \} \qquad \text{for } (i,j) \in [n^a] \times [n^c].
\]

\paragraph*{Min-Plus Convolution.}
Let $A$ and $B$ be arrays of length $n$ indexed from $1$. The \emph{Min-Plus convolution} of $A$ and $B$ is an array $A \diamond B$ of length $2n-1$, indexed from $2$ to $2n$, defined by
\[
    (A \diamond B)_k := \min_{\substack{1 \le i \le k-1 \\ 1 \le i, k - i \le n} } \{ A_i + B_{k-i} \},
\qquad \text{for } k \in [2, 2n].
\]

\paragraph*{Monotonicity.}
An $n \times m$ matrix $A$ is \emph{row-monotone} if for every $i \in [n]$ and $j \in [m-1]$, $A_{i,j} \le A_{i,j+1}$.\footnote{Unlike in the introduction, where boundedness is included for readability, monotonicity here only refers to the ordering condition; boundedness is stated separately.} It is \emph{column-monotone} if for every $i \in [n-1]$ and $j \in [m]$, $A_{i,j} \le A_{i+1,j}$. An array $A$ of length $n$ is \emph{monotone} if for every $i \in [n-1]$, $A_i \le A_{i+1}$.

\paragraph*{Rectangular Matrix Multiplication.}
For $a,b,c \ge 0$, let $\omega(a,b,c)$ denote an exponent such that the product of an $n^a \times n^b$ matrix and an $n^b \times n^c$ matrix can be computed in $\hO(n^{\omega(a,b,c)})$ arithmetic operations. The best known upper bound is $\omega(1,1,1) < 2.372$ \cite{AlmanDWXXZ25}. It is known that $\omega(a, b, c) = \omega(a, c, b)$~\cite{DBLP:journals/jc/HuangP98}. 

\paragraph*{Polynomial Multiplication.}
Let $P,Q \in \mathbb{Z}[x,y]$ be bivariate polynomials whose degree in $x$ is at most $d_x$ and whose degree in $y$ is at most $d_y$. We can compute $P+Q$ and $P-Q$ in $\cO(d_x d_y)$ time, and compute $P \cdot Q$ in $\tcO(d_x d_y)$ time using FFT. Let $P(x) \in (\mathbb{Z}[x])^{n^a \times n^b}$ and $Q(x) \in (\mathbb{Z}[x])^{n^b \times n^c}$ be polynomial matrices whose entries have degree at most $D$. Then the matrix product $P(x)\cdot Q(x)$ can be computed in $\hO(D\, n^{\omega(a,b,c)})$ time.

\paragraph*{Prime Number Theorem.}
Let $\pi(x)$ denote the number of primes at most $x$. The \emph{prime number theorem} states that $ \pi(x) =  x/\log x \pm o(x / \log x)$. As a corollary, for any sufficiently large integer $R$, the number of primes in the interval $[R/2, R]$ is $\pi(R) - \pi(R/2 - 1) = \Theta\left(R/\log R\right)$.

%% file: product.tex
\section{Deterministic Monotone Min-Plus Product}
\label{sec:product}

In this section, we present a deterministic algorithm for row-monotone Min-Plus product, whose definition we recall below.
\begin{prob}
\label{prob:monotone-min-plus}
    Let $A$ be an $n^a \times n^b$ integer matrix. Let $B$ be an $n^b \times n^c$ row-monotone integer matrix with entries in $[n^\mu]$. Given $(A,B)$ as input, compute $A \star B$.
\end{prob}

\subsection{Reduction to a Verification Problem}
\label{subsec:reductions}
We start with a standard reduction that transforms the problem into a simpler form. Consider the following verification problem.

\begin{prob}
\label{prob:exact-tri}
    Let $A$ be an $n^a \times n^b$ integer matrix with nonnegative entries bounded by $\cO(n^\mu)$. Let $B$ be an $n^b \times n^c$ row-monotone integer matrix with nonnegative entries bounded by $\cO(n^\mu)$. Let $C$ be an $n^a \times n^c$ row-monotone integer matrix with nonnegative entries bounded by $\cO(n^\mu)$. Given $(A,B,C)$ as input, for every $(i,j) \in [n^a] \times [n^c]$, decide whether there exists $k \in [n^b]$ such that $A_{i,k} + B_{k,j} = C_{i,j}$.
\end{prob}

The reduction uses the following simple fact.
\begin{fact}[\cite{CDXZ22}]
\label{fact:row-monotone-C}
If the rows of $B$ are monotone, then the rows of $A \star B$ are monotone.
\end{fact}
\begin{proof}
For any $(i,j) \in [n^a] \times [n^c - 1]$, we have
$$
    (A \star B)_{i,j} = \min_{k} \{A_{i,k}+ B_{k,j}\}
                      \le \min_{k} \{A_{i,k}+ B_{k,j + 1}\}
                       = (A \star B)_{i,j + 1}.
$$
\end{proof}

\begin{lemma}
\label{lem:reduction-to-verification}
If there is a deterministic algorithm for \cref{prob:exact-tri} that runs in $T(n^a, n^b, n^c; n^{\mu})$ time, then there is a deterministic algorithm for \cref{prob:monotone-min-plus} that runs in $\tcO(T(n^a, n^b, n^c; n^{\mu}))$ time.
\end{lemma}
\begin{proof}
First, by a simple observation of~\cite{CDXZ22}, we may assume that all entries of $A$ are nonnegative integers bounded by $\cO(n^\mu)$: for each $i$, subtract $\delta_i$ from every entry in the $i$-th row of $A$, where $\delta_i$ is the minimum value in that row. Then any entry of $A$ greater than $2 n^\mu$ can be set to $2 n^\mu+1$, since it can never participate in an optimal solution. After computing the Min-Plus product of this modified matrix $A$ and $B$, we simply add $\delta_i$ back to the $i$-th row of the result.  

Let $A'_{i,j} := \lfloor A_{i,j}/2 \rfloor$ and $B'_{i,j} := \lfloor B_{i,j}/2 \rfloor$. Flooring preserves row-monotonicity, so the rows of $B'$ are also row-monotone. We compute $C' := A' \star B'$ recursively. The base case occurs when $A$ and $B$ are zero matrices, so we can compute $A \star B$ trivially. Now define three $n^a \times n^c$ candidate matrices $C^{(0)}, C^{(1)}, C^{(2)}$ by 
\[ C^{(s)}_{i,j} := 2C'_{i,j} + s \qquad \text{for }s \in \{0,1,2\}.\]
Since the rows of $B'$ are monotone, the rows of $C'$ are also monotone by \cref{fact:row-monotone-C}. Therefore, the rows of each $C^{(s)}$ are monotone as well, and we can run the algorithm for \cref{prob:exact-tri} on the instance $(A,B,C^{(s)})$ for each $s \in \{0,1,2\}$.

Since $(x-1)/2 \le \lfloor x/2 \rfloor \le x/2$ and $C' = A' \star B'$, we have
\[
    C'_{i,j}
    = \min_{k} \left\{  \left\lfloor\frac{A_{i,k}}{2}\right\rfloor + \left\lfloor\frac{B_{k,j}}{2}\right\rfloor \right\}
    \le \min_k \left\{\frac{A_{i,k}+B_{k,j}}{2}\right\}
    = \frac{(A \star B)_{i,j}}{2},
\]
and
\[
    C'_{i,j} 
    = \min_{k} \left\{  \left\lfloor\frac{A_{i,k}}{2}\right\rfloor + \left\lfloor\frac{B_{k,j}}{2}\right\rfloor \right\} 
    \ge \min_{k} \left\{  \frac{A_{i,k} +B_{k,j} -2}{2} \right\}
    =  \frac{(A \star B)_{i,j}-2}{2}. 
 \]
Together, these inequalities imply $(A \star B)_{i,j} - 2C'_{i,j} \in \{0,1,2\}$. Thus, for each fixed pair $(i,j)$, there exists at least one $s \in \{0,1,2\}$ such that $C^{(s)}_{i,j} = (A \star B)_{i,j}$, and we can recover $(A \star B)_{i,j}$ by taking the minimum value among $C^{(s)}_{i,j}$ for which the verifier accepts $(i,j)$ on $(A,B,C^{(s)})$. Moreover, if the verifier accepts $(i,j)$ on $(A,B,C^{(s)})$ with witness $k$, then $C^{(s)}_{i,j} = A_{i,k}+B_{k,j} \ge (A \star B)_{i,j}$. Therefore, taking the minimum accepted value among the $C^{(s)}_{i,j}$ returns exactly $(A \star B)_{i,j}$.

For the running time, we make $\cO(1)$ calls to the $T(n^a, n^b, n^c; n^{\mu})$-time verifier per recursion level.
The recursion depth is $\cO(\log(n^\mu)) = \cO(\log n)$.
The additional costs of forming $A'$, $B'$, and the matrices $C^{(s)}$ are dominated by $T(n^a, n^b, n^c; n^\mu)$. Hence, the total running time is $\tcO(T(n^a, n^b, n^c; n^{\mu}))$.
\end{proof}

Using an approach similar to that of~\cite[Lemma~3.5]{CDXZ22}, we further reduce \cref{prob:exact-tri} to the following problem and focus on solving it efficiently.

\begin{prob}
\label{prob:exact-tri-mod-M}
Let $(A,B,C)$ be an instance of \cref{prob:exact-tri}, and let $M$ be a positive integer that is a multiple of $100$ and satisfies $M = \cO(n^{\mu})$. Additionally, assume that for every $i \in [n^a]$, $k \in [n^b]$, and $j \in [n^c]$, we have
$$
    A_{i,k} \bmod M,\; B_{k,j} \bmod M,\; C_{i,j} \bmod M \le \frac{M}{10}.
$$
Given $(A,B,C)$ as input, for every $(i,j) \in [n^{a}] \times [n^{c}]$, decide whether there exists $k \in [n^{b}]$ such that $A_{i,k} + B_{k,j} = C_{i,j}$.
\end{prob}

\begin{lemma}
\label{lem:reduction-to-small-modM}
    If there is a $T(n^a, n^b, n^c; n^{\mu}, M)$-time algorithm for \cref{prob:exact-tri-mod-M} for any such $M$, then there is an $\tcO(T(n^a, n^b, n^c; n^{\mu}, M))$-time algorithm for \cref{prob:exact-tri}.
\end{lemma}
\begin{proof}
Let $W := M/100$, which is an integer since $M$ is a multiple of $100$. Without loss of generality, we may assume that all entries of $A$, $B$, and $C$ are at least $M$ by adding $M$ to all entries of $A$ and $B$ and adding $2M$ to all entries of $C$; this preserves witnesses. Note that $2M = \cO(n^{\mu})$, so the entries remain bounded by $\cO(n^{\mu})$.

We partition $\{0,1,\ldots,M-1\}$ into intervals $I_s$ defined by
$$
    I_s := [sW,(s+1)W) \qquad \text{for } s \in \{0,1,\ldots,99\}.
$$
Then, for each $s \in \{0,1,\ldots,99\}$, construct an $n^a \times n^b$ matrix $A^{(s)}$ and an $n^{b} \times n^c$ matrix $B^{(s)}$ by
$$
A^{(s)}_{i,k} :=
\begin{cases}
A_{i,k} - sW, & \text{if } A_{i,k} \bmod M \in I_s,\\
\Big\lfloor \dfrac{A_{i,k}-sW}{M}\Big\rfloor \cdot M + 3W, & \text{otherwise,}
\end{cases}
$$
and
$$
B^{(s)}_{k,j} :=
\begin{cases}
B_{k,j} - sW, & \text{if } B_{k,j} \bmod M \in I_s,\\
\Big\lfloor \dfrac{B_{k,j}-sW}{M}\Big\rfloor \cdot M + 3W, & \text{otherwise.}
\end{cases}
$$
Additionally, for all $s,t \in \{0,1,\ldots,99\}$, let
$$
J_{s,t} := \bigl\{(s+t)W \bmod M,\; ((s+t)W+1) \bmod M,\; \ldots,\; ((s+t)W + (2W-1))\bmod M \bigr\},
$$
and construct $n^a \times n^c$ matrix $C^{(s,t)}$ by
$$
C^{(s,t)}_{i,j} :=
\begin{cases}
C_{i,j} - (s + t)W, & \text{if } C_{i,j} \bmod M \in J_{s,t},\\
\Big\lfloor \dfrac{C_{i,j}-(s+t)W}{M}\Big\rfloor \cdot M + 7W, & \text{otherwise.}
\end{cases}
$$
Now, for all $s,t \in \{0,1,\ldots,99\}$, we run the algorithm for \cref{prob:exact-tri-mod-M} on the $100^2 = \cO(1)$ instances $(A^{(s)}, B^{(t)}, C^{(s,t)})$ and return \textsc{Yes} for $(i, j)$ if any instance returns \textsc{Yes} for $(i, j)$.

We first justify that our instances satisfy the promise of \cref{prob:exact-tri-mod-M}. Fix $s,t \in \{0,1,\ldots,99\}$. It is clear that all entries of $A^{(s)}$, $B^{(t)}$, and $C^{(s,t)}$ are nonnegative, since $sW$ and $tW$ are at most $99M/100 < M$ and all entries of $A$, $B$, and $C$ are at least $M$. For the small-residue requirement modulo $M$, note that if $A_{i,k} \bmod M \in I_s$, then
$$
    0 \le A^{(s)}_{i,k}\bmod M < W < M/10,
$$
and if $A_{i,k} \bmod M \notin I_s$, then
$$
    0 \le A^{(s)}_{i,k}\bmod M = 3W < M/10.
$$
Similarly, the same residue bounds modulo $M$ hold for $B^{(t)}$.
For the residue bound for $C^{(s,t)}$, note that if $C_{i,j} \bmod M \in J_{s,t}$, then
$$
    0 \le C^{(s,t)}_{i,j}\bmod M < 2W < M/10,
$$
and otherwise
$$
    C^{(s,t)}_{i,j}\bmod M = 7W < M/10.
$$
Hence, all three matrices have residues modulo $M$ bounded above by $M/10$. For the monotonicity constraint, we prove only the row-monotonicity of $B^{(t)}$; the row-monotonicity of $C^{(s,t)}$ follows analogously.
Fix $k \in [n^b]$ and $j \in [n^c-1]$ and write
$$
B_{k,j} - tW= q_j M + r_j,\qquad B_{k,j+1} -tW = q_{j+1} M + r_{j+1},
\qquad \text{where } r_j,r_{j+1}\in\{0,1,\ldots,M-1\}.
$$
Then we can rewrite the entries of $B^{(t)}$ as
$$
B^{(t)}_{k,j}=
\begin{cases}
q_jM+r_j, & \text{if } r_j\in[0,W),\\
q_jM+3W, & \text{if } r_j\notin[0,W),
\end{cases}
\qquad
B^{(t)}_{k,j + 1}=
\begin{cases}
q_{j+1}M+r_{j+1}, & \text{if } r_{j+1}\in[0,W),\\
q_{j+1}M+3W, & \text{if } r_{j+1}\notin[0,W).
\end{cases}
$$
By the monotonicity of $B$, we have two cases:
\begin{enumerate}
    \item If $q_j<q_{j+1}$, using the fact that $3W<M$, we have
    $$
        B^{(t)}_{k,j}\le q_jM+3W < (q_j+1)M \le q_{j+1}M \le B^{(t)}_{k,j+1}.
    $$
    \item Otherwise, we have $q_j=q_{j+1}$ and $r_j\le r_{j+1}$.
    If both $r_j,r_{j+1} \in [0,W)$ or both $r_j,r_{j+1} \notin [0,W)$, then $B^{(t)}_{k,j} \le B^{(t)}_{k,j+1}$ is immediate. If $r_j\in[0,W)$ but $r_{j+1}\notin[0,W)$, then
    $$
        B^{(t)}_{k,j} = q_jM+r_j < q_jM+3W = q_{j+1}M+3W = B^{(t)}_{k,j+1}.
    $$
\end{enumerate}
Therefore, each row of $B^{(t)}$ is row-monotone, as desired.

\begin{claim}[Correctness of the reduction]
\label{clm:reduction-correctness}
Fix $(i,j)\in[n^a]\times[n^c]$. We have that there exists $k\in[n^b]$ such that $A_{i,k}+B_{k,j}=C_{i,j}$ if and only if there exist $s,t\in\{0,1,\ldots,99\}$ and $k\in[n^b]$ such that
$A^{(s)}_{i,k}+B^{(t)}_{k,j}=C^{(s,t)}_{i,j}$.
\end{claim}
\begin{proof}
Fix $(i,j)\in[n^a]\times[n^c]$.
Suppose there exists $k\in[n^b]$ such that $A_{i,k}+B_{k,j}=C_{i,j}$.
Let $s,t\in\{0,1,\ldots,99\}$ be such that $A_{i,k}\bmod M\in I_{s}$ and $B_{k,j}\bmod M\in I_{t}$. Then
\begin{equation}
    \label{eq:reduction-forward}
    A^{(s)}_{i,k} + B^{(t)}_{k,j}= A_{i,k}+B_{k,j}-(s+t)W = C_{i,j}-(s+t)W.
\end{equation}
Moreover,
$$
    0 \le  A^{(s)}_{i,k} \bmod M < W,\qquad
    0 \le   B^{(t)}_{k,j} \bmod M < W,
$$
so the sum does not wrap around modulo $M$, which implies
$$
    0 \le \left(A^{(s)}_{i,k} + B^{(t)}_{k,j} \right) \bmod M < 2W.
$$
Then the equality \cref{eq:reduction-forward} gives
$$
    0 \le \left(C_{i,j}-(s+t)W \right)\bmod M < 2W,
$$
which implies $C_{i,j}\bmod M \in J_{s,t}$; thus $C^{(s,t)}_{i,j}=C_{i,j} - (s + t)W$. Hence, we have $A^{(s)}_{i,k} + B^{(t)}_{k,j} = C^{(s,t)}_{i,j}$ by \cref{eq:reduction-forward}, and the instance corresponding to $(s,t)$ verifies the triangle with witness $k$.

Conversely, suppose that for some $s,t\in\{0,1,\ldots,99\}$, the instance has a witness $k \in [n^b]$ such that
\begin{equation}
\label{eq:reduction-backward}
A^{(s)}_{i,k} + B^{(t)}_{k,j}
= C^{(s,t)}_{i,j}.
\end{equation}
We claim that $A_{i,k}\bmod M\in I_{s}$ and $B_{k,j}\bmod M\in I_{t}$.
If $A_{i,k}\bmod M\notin I_{s}$ or $B_{k,j}\bmod M\notin I_{t}$, then
$$
    \bigl(A^{(s)}_{i,k} + B^{(t)}_{k,j} \bigr)\bmod M \in [3W, 6W],
$$
since at least one summand has residue $3W$ and the other has residue in $[0,W)\cup\{3W\}$, and there is no wraparound modulo $M$. On the other hand, we have
$$
    C^{(s,t)}_{i,j}\bmod M \in [0,2W) \cup \{7W\},
$$
contradicting \cref{eq:reduction-backward}. Hence, $A_{i,k}\bmod M\in I_{s}$ and $B_{k,j}\bmod M\in I_{t}$. Moreover, under this condition, the left-hand side residue in \cref{eq:reduction-backward} lies in $[0,2W)$, so \cref{eq:reduction-backward} forces $C^{(s,t)}_{i,j} \bmod M \in [0,2W)$.
Then, $C_{i,j}\bmod M\in J_{s,t}$, and we have $C^{(s,t)}_{i,j}=C_{i,j} - (s + t)W$. Therefore,
$$
    A_{i,k}+B_{k,j} - (s + t)W = A^{(s)}_{i,k} + B^{(t)}_{k,j} = C^{(s,t)}_{i,j} = C_{i,j}  - (s + t)W,
$$
and the correctness follows.
\end{proof}
\end{proof}

\subsection{Algorithm for \texorpdfstring{\cref{prob:exact-tri-mod-M}}{Problem \ref*{prob:exact-tri-mod-M}}}

Suppose we are given an instance $(A,B,C)$ of \cref{prob:exact-tri-mod-M} with $M = \Theta(n^{d})$ for some constant $0 \le d \le \mu$. Define
\[
    \Ah_{i,k}=\left\lfloor \frac{A_{i,k}}{M}\right\rfloor,
    \qquad
    \Bh_{k,j}=\left\lfloor \frac{B_{k,j}}{M}\right\rfloor,
    \qquad
    \Ch_{i,j}=\left\lfloor \frac{C_{i,j}}{M}\right\rfloor.
\]
We also define
\[
    \Al_{i,k}=A_{i,k}\bmod M,
    \qquad
    \Bl_{k,j}=B_{k,j}\bmod M,
    \qquad
    \Cl_{i,j}=C_{i,j}\bmod M.
\]
The high-level outline of our algorithm is as follows:
\begin{enumerate}
    \item Find a ``good'' modulus $Q$ with $M \le Q \le Mn^{o(1)}$. The definition of ``good'' is deferred to \cref{def:good-modulus}. 
    \item For every $(i,j)\in[n^a]\times[n^c]$, compute $s_{i,j}$, the number of $k\in[n^b]$ such that
    $A_{i,k}+B_{k,j}\equiv C_{i,j}\pmod{Q}$.
    \item For every $(i,j)\in[n^a]\times[n^c]$, compute $s'_{i,j}$, the number of $k\in[n^b]$ such that
    $A_{i,k}+B_{k,j}\equiv C_{i,j}\pmod{Q}$ and $\Ah_{i,k}+\Bh_{k,j}\ne \Ch_{i,j}$.
    \item For each $(i,j)\in[n^a]\times[n^c]$, output \textsc{Yes} if and only if $s_{i,j}>s'_{i,j}$.
\end{enumerate}
We will prove that $A_{i, k} + B_{k, j} = C_{i, j}$ if and only if $A_{i, k} + B_{k, j} \equiv C_{i, j} \pmod{Q}$ and $\Ah_{i, k} + \Bh_{k, j} = \Ch_{i, j}$. Hence, there exists a witness $k$ if and only if $s_{i,j} > s'_{i,j}$. We first prove two simple facts that bound $\lvert A_{i,k}+B_{k,j}-C_{i,j}\rvert$, depending on whether $\Ah_{i,k}+\Bh_{k,j}=\Ch_{i,j}$, and then prove the correctness of the algorithm assuming that our subroutines are correct.

\begin{fact}
\label{fact:big-difference}
For every $(i,k,j)\in[n^a]\times[n^b]\times[n^c]$, if $\Ah_{i,k}+\Bh_{k,j}\ne \Ch_{i,j}$, then $|A_{i,k}+B_{k,j}-C_{i,j}|\ge 7M/10$.
\end{fact}
\begin{proof}
We decompose $A_{i,k}+B_{k,j}-C_{i,j}
= M\cdot\bigl(\Ah_{i,k}+\Bh_{k,j}-\Ch_{i,j}\bigr) + \bigl(\Al_{i,k}+\Bl_{k,j}-\Cl_{i,j}\bigr)$. If $\Ah_{i,k}+\Bh_{k,j}-\Ch_{i,j}\ne 0$, then $|\Ah_{i,k}+\Bh_{k,j}-\Ch_{i,j}|\ge 1$. By the promise of \cref{prob:exact-tri-mod-M}, $|\Al_{i,k}+\Bl_{k,j}-\Cl_{i,j}| \le |\Al_{i,k}|+|\Bl_{k,j}|+|\Cl_{i,j}| \le 3M/10$. Therefore, using $|a+b|\ge |a|-|b|$, we have
\[
    \bigl|A_{i,k}+B_{k,j}-C_{i,j}\bigr|
    \ge M\cdot\bigl|\Ah_{i,k}+\Bh_{k,j}-\Ch_{i,j}\bigr| - \bigl|\Al_{i,k}+\Bl_{k,j}-\Cl_{i,j}\bigr|
    \ge M-3M/10
    =7M/10.
\]
\end{proof}

\begin{fact}
\label{fact:small-difference}
For every $(i,k,j)\in[n^a]\times[n^b]\times[n^c]$, if $\Ah_{i,k}+\Bh_{k,j}=\Ch_{i,j}$, then $|A_{i,k}+B_{k,j}-C_{i,j}|\le 3M/10$.
\end{fact}
\begin{proof}
Since $\Ah_{i,k}+\Bh_{k,j}-\Ch_{i,j}=0$, we have
$|A_{i,k}+B_{k,j}-C_{i,j}|= |\Al_{i,k}+\Bl_{k,j}-\Cl_{i,j}|
\le 3M/10$ by the same argument as in \cref{fact:big-difference}.
\end{proof}

Recall that $s_{i,j}$ is the number of $k\in[n^b]$ such that $A_{i,k}+B_{k,j}\equiv C_{i,j}\pmod Q$, and $s'_{i,j}$ is the number of $k\in[n^b]$ such that $A_{i,k}+B_{k,j}\equiv C_{i,j}\pmod Q$ and $\Ah_{i,k}+\Bh_{k,j}\ne \Ch_{i,j}$. We now prove the correctness of our algorithm.

\begin{lemma}[Correctness]
\label{lem:s-vs-sprime}
For every $(i,j)\in[n^a]\times[n^c]$, there exists $k\in[n^b]$ such that $A_{i,k}+B_{k,j}=C_{i,j}$ if and only if $s_{i,j}>s'_{i,j}$.
\end{lemma}

\begin{proof}
Fix $(i,j)$ and suppose there exists $k\in[n^b]$ such that $A_{i,k}+B_{k,j}=C_{i,j}$. Then $A_{i,k}+B_{k,j}\equiv C_{i,j}\pmod Q$, so this $k$ contributes one to $s_{i,j}$. Moreover, under the promise of \cref{prob:exact-tri-mod-M}, we have $\Al_{i,k}+\Bl_{k,j}\le 2M/10$, so adding the low parts creates no carry across a multiple of $M$. Hence,
$\Ch_{i,j} =\lfloor C_{i,j}/M \rfloor =\lfloor (A_{i,k}+B_{k,j})/M\rfloor =\Ah_{i,k}+\Bh_{k,j}$.
Therefore, this $k$ does not contribute to $s'_{i,j}$, and so $s_{i,j}>s'_{i,j}$.

Conversely, suppose $s_{i,j}>s'_{i,j}$. Then there exists $k\in[n^b]$ such that $A_{i,k}+B_{k,j}\equiv C_{i,j}\pmod Q$ and $\Ah_{i,k}+\Bh_{k,j}=\Ch_{i,j}$. By \cref{fact:small-difference}, we have $|A_{i,k}+B_{k,j}-C_{i,j}|\le 3M/10 < Q$. Since $A_{i,k}+B_{k,j}-C_{i,j}$ is a multiple of $Q$ and has magnitude strictly less than $Q$, it must be zero. Therefore, $A_{i,k}+B_{k,j}=C_{i,j}$.
\end{proof}

In the following sections, we describe how to compute $s_{i,j}$ and $s'_{i,j}$ efficiently, and how to choose and find a modulus $Q$. To this end, we introduce several objects used by the subroutines and their analysis. We first provide some intuition. As in \cite{CDXZ22}, we view each pair of row-monotone rows $B_{k,*}$ and $C_{i,*}$ through a level-$\ell$ ``coarse lens,'' by rounding entries down to multiples of $2^\ell$. This partitions the column indices into contiguous blocks on which both rounded rows are constant; these blocks are called \emph{level-$\ell$ segments}. Within a segment, the values $B_{k,j}$ and $C_{i,j}$ vary by at most $2^\ell$, so checking a modular relation at the segment start (namely, $A_{i, k} + B_{k, j_0} \equiv C_{i, j_0} \pmod{Q}$, where $j_0$ is the segment start) represents the entire segment up to an additive slack of $\pm 4\cdot 2^{\ell}$. We call a segment \emph{active} (with respect to a modulus $Q$) if the corresponding modular relation can hold for some shift in $[-4\cdot 2^\ell,\,4\cdot 2^\ell]$ while the high parts indicate that equality does not hold. Thus, active segments capture the remaining ``spurious'' modular matches. Finally, a modulus $Q$ is \emph{good} if, for every level $\ell$, it leaves only a few active segments so that we can afford to enumerate and eliminate all spurious matches.

More formally, let $\lm$ be such that $M/20 \le 2^{\lm} < M/10$.

\begin{definition}[Segments]
\label{def:segments}
For $0 \le \ell \le \lm$, a level-$\ell$ segment is a triple $(i, k, [j_0, j_1])$ such that for every $j \in [j_0, j_1]$,
$\lfloor B_{k, j_0} / 2^\ell\rfloor = \lfloor B_{k, j} / 2^\ell\rfloor$ and $\lfloor C_{i, j_0} / 2^\ell\rfloor = \lfloor C_{i, j} / 2^\ell\rfloor$, and $[j_0, j_1]$ cannot be extended further.
\end{definition}

\begin{fact}
\label{fact:num-segments}
For $0 \le \ell \le \lm$, the total number of level-$\ell$ segments is $\cO(n^{a+b+\mu}/2^\ell)$.
\end{fact}
\begin{proof}
Fix $(i,k)\in[n^a]\times[n^b]$. For each $j\in[n^c]$, define
$b_j:=\lfloor B_{k,j}/2^\ell\rfloor$ and $c_j:=\lfloor C_{i,j}/2^\ell\rfloor$. Both sequences $(b_j)_j$ and $(c_j)_j$ are monotone. Since all entries of $B$ and $C$ are bounded by $\cO(n^\mu)$, each of $(b_j)_j$ and $(c_j)_j$ is bounded by $\cO(n^\mu/2^\ell)$, and therefore each has at most $\cO(n^\mu/2^\ell)$ constant blocks. A level-$\ell$ segment is a maximal interval on which both sequences are constant; thus the number of level-$\ell$ segments is $\cO(n^\mu/2^\ell)$. Summing over all $n^a n^b$ choices of $(i,k)$ gives $\cO(n^{a+b+\mu}/2^\ell)$ segments in total.
\end{proof}

\begin{definition}[Active segment]
\label{def:active-segment}
For $0 \le \ell \le \lm$, a level-$\ell$ segment $(i,k,[j_0,j_1])$ is called \emph{active} with respect to $Q$ if $\Ah_{i,k}+\Bh_{k,j_0} \ne \Ch_{i,j_0}$ and there exists $s\in [-4\cdot 2^\ell,4\cdot 2^\ell]$ such that $A_{i,k}+B_{k,j_0}-C_{i,j_0}\equiv s \pmod Q$.
\end{definition}

\begin{definition}[Set of active segments]
\label{def:active-segments}
For $0 \le \ell \le \lm$, let $S_{\ell}(Q)$ be the set of all level-$\ell$ segments that are active with respect to $Q$.
\end{definition}

\begin{definition}[Good modulus]
\label{def:good-modulus}
We say that an integer $Q$ is a \emph{good modulus} if $M \le Q \le Mn^{o(1)}$ and $|S_{\ell}(Q)| \le \hO(n^{a+b+\mu}/Q)$ for all $0 \le \ell \le \lm$.
\end{definition}

\subsection{Computing \texorpdfstring{$s_{i,j}$ and $s'_{i,j}$}{si,j and s'i,j}}
Recall that $s_{i,j}$ is the number of indices $k\in[n^b]$ such that $A_{i,k}+B_{k,j}\equiv C_{i,j}\pmod{Q}$. We compute $s_{i,j}$ for all $(i,j) \in [n^a] \times [n^c]$ using polynomial matrix multiplication, which reduces modular counting to ordinary rectangular matrix multiplication.

\begin{lemma}[Computing $s_{i,j}$]
\label{lem:compute-s}
We can compute $s_{i,j}$ for every $(i,j) \in [n^a] \times [n^c]$ in $\hO(Q n^{\omega(a,b,c)})$ time.
\end{lemma}

\begin{proof}
Work over the ring $R:=\mathbb{F}[x]/(x^Q-1)$, where $\mathbb{F}$ has characteristic larger than $n^b$.
Define an $n^a\times n^b$ matrix $A'$ and an $n^b\times n^c$ matrix $B'$ over $R$ by
$$
    A'_{i,k}(x):=x^{A_{i,k}\bmod Q},
    \qquad
    B'_{k,j}(x):=x^{B_{k,j}\bmod Q}.
$$
Compute the matrix product $P(x):=A'(x)\cdot B'(x)$ over $R$ using fast rectangular matrix multiplication. For each residue $r\in\{0,1,\ldots,Q-1\}$, the coefficient of $x^r$ in $P_{i,j}(x)$ is the number of indices $k\in[n^b]$ such that $A_{i,k}+B_{k,j}\equiv r\pmod Q$. Hence, we set $s_{i,j}$ to be the coefficient of $x^{C_{i,j}\bmod Q}$ in $P_{i,j}(x)$. Since multiplication in $R$ reduces exponents modulo $Q$, multiplying two entries costs $\hO(Q)$ field operations, and the overall matrix multiplication takes $\hO(Qn^{\omega(a,b,c)})$ time.
\end{proof}

Recall that $s'_{i,j}$ is the number of indices $k$ such that $A_{i,k}+B_{k,j}\equiv C_{i,j}\pmod Q$ and $\Ah_{i,k}+\Bh_{k,j}\ne \Ch_{i,j}$. Suppose that $Q$ is a good modulus as defined in \cref{def:good-modulus}; we show how to find such a modulus in the next section. For simplicity, let $S_{\ell}$ denote the set of all active level-$\ell$ segments, suppressing the dependence on $Q$. At a high level, we compute $s'_{i,j}$ for all $(i,j)$ using $S_0$. Indeed, since $B_{k,j}$ and $C_{i,j}$ are constant on a level-$0$ segment, the congruence $A_{i,k}+B_{k,j}\equiv C_{i,j}\pmod Q$ can be tested at a single $j$ and then applied to the entire interval $[j_0, j_1]$ in a segment. To compute $S_0$, we compute the sets $S_{\ell}$ via a top-down refinement from level $\lm$ to level $0$. To formalize our idea, we prove a property on which our algorithm relies.

\begin{lemma}
\label{lem:nesting-active}
    For all $0 \le \ell < \lm$, and $(i,k,[j_0', j_1']) \in S_{\ell}$, the unique level-$(\ell + 1)$ segment $(i,k,[j_0, j_1])$ with $[j_0', j_1'] \subseteq [j_0, j_1]$ is active, i.e., $(i,k,[j_0, j_1]) \in S_{\ell + 1}$.
\end{lemma}
\begin{proof}
    Fix $j \in [j_0,j_1]$. By definition of a level-$(\ell+1)$ segment, we have $\lfloor B_{k,j_0}/2^{\ell+1}\rfloor = \lfloor B_{k,j}/2^{\ell+1}\rfloor$. Hence, $|B_{k,j}-B_{k,j_0}|<2^{\ell+1}< M/10$. Moreover, $|\Bl_{k,j}-\Bl_{k,j_0}| \le M/10$ by the promise of \cref{prob:exact-tri-mod-M}. Hence, if $\Bh_{k,j}\neq \Bh_{k,j_0}$, then
    \[
        |B_{k,j}-B_{k,j_0}| = |(\Bh_{k,j} - \Bh_{k,j_0})M + (\Bl_{k,j} - \Bl_{k,j_0})|
        \ge M-|\Bl_{k,j}-\Bl_{k,j_0}|
        \ge 9M/10,
    \] 
    which contradicts the previous bound. Therefore, $\Bh_{k,j}$ is constant for $j \in [j_0,j_1]$, and by a similar argument, $\Ch_{i,j}$ is constant for $j \in [j_0,j_1]$. Now, since $(i,k,[j_0', j_1']) \in S_{\ell}$, we have $\Ah_{i,k}+\Bh_{k,j_0'}-\Ch_{i,j_0'}\neq 0$. Because $j_0' \in [j_0, j_1]$, we conclude that $\Ah_{i,k}+\Bh_{k,j_0}-\Ch_{i,j_0}\neq 0$, satisfying the high-part condition of being active.

    For the congruence condition, note that $[j_0',j_1'] \subseteq [j_0,j_1]$ and $(i,k,[j_0,j_1])$ is a level-$(\ell+1)$ segment. Hence, $|B_{k,j_0}-B_{k,j_0'}|<2^{\ell+1}$ and $|C_{i,j_0}-C_{i,j_0'}|<2^{\ell+1}$, which implies
    \[
        |(A_{i,k}+B_{k,j_0}-C_{i,j_0})-(A_{i,k}+B_{k,j_0'}-C_{i,j_0'})|< 2\cdot 2^{\ell+1}.
    \]
    Since $(i,k,[j_0',j_1']) \in S_{\ell}$, we have $A_{i,k}+B_{k,j_0'}-C_{i,j_0'}\equiv s'\pmod Q$ for some $s' \in [-4\cdot 2^\ell,4\cdot 2^\ell]$. Take
    \[
        s = s' + (A_{i,k}+B_{k,j_0}-C_{i,j_0})-(A_{i,k}+B_{k,j_0'}-C_{i,j_0'}).
    \]
    Then $|s|<4\cdot 2^\ell + 2 \cdot 2^{\ell+1}=4\cdot 2^{\ell+1}$ and
    \[
        (A_{i,k}+B_{k,j_0}-C_{i,j_0})
        \equiv s - s' +  (A_{i,k}+B_{k,j_0'}-C_{i,j_0'})
        \equiv s \pmod Q.
    \]
    Thus, the level-$(\ell+1)$ segment $(i,k,[j_0,j_1])$ is active.
\end{proof}

\begin{lemma}[Computing $S_0$]
\label{lem:compute-S}
We can compute $S_0$ in $\hO(n^{a+b+\mu}/Q)$ time.
\end{lemma}
\begin{proof}
We first compute all level-$\lm$ segments as follows.
For each $(i,k) \in [n^a] \times [n^b]$, we use binary search on $j$ to group maximal consecutive indices on which
$\lfloor B_{k,j}/2^{\lm}\rfloor$ and $\lfloor C_{i,j}/2^{\lm}\rfloor$ are both constant. We then enumerate all those level-$\lm$ segments and filter the active ones to
obtain $S_{\lm}$. Next, for $\ell=\lm-1,\ldots,0$, we construct $S_{\ell}$ from $S_{\ell+1}$ as follows. Every segment in $S_{\ell+1}$ refines into $\cO(1)$ level-$\ell$ subsegments, since decreasing $\ell$ by $1$ can split a maximal constant block into at most two blocks for $B$ and at most two blocks for $C$. We use binary search to compute these subsegments. For each resulting subsegment, we test whether it is active at level $\ell$ and, if so, insert it into $S_{\ell}$. 

For correctness, note that every active level-$\ell$ segment is a subsegment of an active level-$(\ell + 1)$ segment by \cref{lem:nesting-active}, and so the correctness follows inductively.

By \cref{fact:num-segments}, there are $\tcO(n^{a+b+\mu}/2^{\lm})$ level-$\lm$ segments. Thus, computing them via binary search and constructing $S_{\lm}$ takes $\tcO(n^{a+b+\mu}/2^{\lm}) = \tcO(n^{a+b+\mu}/M) = \hO(n^{a+b+\mu}/Q)$ time. For the refinement, we note that $Q$ is a good modulus, so for each $\ell$, we have $|S_{\ell}|= \hO(n^{a+b+\mu}/Q)$, and each active segment generates only $\cO(1)$ subsegments. There are $\cO(\log M) = \tcO(1)$ levels, so the total time spent in this refinement step is $\hO(n^{a+b+\mu}/Q)$.
\end{proof}

\begin{lemma}[Computing $s'_{i,j}$]
\label{lem:compute-sprime}
We can compute $s'_{i,j}$ for all $(i,j)\in[n^a]\times[n^c]$ in time $\hO(n^{a+b+\mu}/Q+n^{a+c})$.
\end{lemma}
\begin{proof}
We first compute the set $S_0$ using \cref{lem:compute-S}.
To compute all $s'_{i,j}$, we aggregate the contributions of active level-$0$ segments in $S_0$ as follows. For each fixed $i\in[n^a]$, we maintain a difference array $D_i$ of length $n^c+1$, indexed by $1,\ldots,n^c+1$, and initialized to zero. For each active level-$0$ segment $(i,k,[j_0,j_1])\in S_0$, we test whether $A_{i,k}+B_{k,j_0}\equiv C_{i,j_0}\pmod Q$. If the test succeeds, we increment $D_i[j_0]$ by one and decrement $D_i[j_1+1]$ by one. After processing all active level-$0$ segments, we compute prefix sums and set
\[
    s'_{i,j}=\sum_{1\le t\le j}D_i[t]\qquad\text{for all } j = 1, \ldots, n^c.
\]

For correctness, fix any triple $(i,k,j)$ counted in $s'_{i,j}$. Then we have $A_{i,k}+B_{k,j}\equiv C_{i,j}\pmod Q$ and $\Ah_{i,k}+\Bh_{k,j}\ne \Ch_{i,j}$. Let $(i,k,[j_0,j_1])$ be the unique level-$0$ segment such that $j \in [j_0, j_1]$. Since $(i,k,[j_0,j_1])$ is a level-$0$ segment, the values $B_{k,j}$ and $C_{i,j}$ are constant over all $j\in[j_0,j_1]$. Since $j \in [j_0, j_1]$, we have 
\[
    \Ah_{i,k}+\Bh_{k,j_0} = \Ah_{i,k}+\Bh_{k,j} \ne \Ch_{i,j} = \Ch_{i,j_0}
\] and 
\[
    A_{i,k}+B_{k,j_0} \equiv A_{i,k}+B_{k,j} \equiv C_{i,j} \equiv C_{i,j_0}\pmod Q.
\] 
Hence, $(i,k,[j_0,j_1])$ is an active level-$0$ segment in $S_0$, and the algorithm performs the range update for $(i,k,[j_0,j_1])$.  We claim that this update increases the reconstructed value $s'_{i,j}$ by one (for this $k$) exactly when $j\in[j_0,j_1]$.
Indeed, for any $j\in[n^c]$,
\begin{itemize}
    \item If $j<j_0$, the prefix sum does not include the increment at position $j_0$, so $s'_{i,j}$ stays the same.
    \item If $j_0\le j\le j_1$, the prefix sum includes the increment at $j_0$ but not the decrement at $j_1+1$, so $s'_{i,j}$ increases by one.
    \item If $j\ge j_1+1$, the prefix sum includes both the increment at $j_0$ and the decrement at $j_1+1$, which cancel, so $s'_{i,j}$ stays the same.
\end{itemize}

Conversely, whenever the algorithm performs a range update for $(i,k,[j_0,j_1])\in S_0$, we have $A_{i,k}+B_{k,j_0}\equiv C_{i,j_0}\pmod Q$. Since $(i,k,[j_0,j_1])$ is a level-$0$ segment, the values $B_{k,j}$ and $C_{i,j}$ are constant over all $j\in[j_0,j_1]$. Therefore, $A_{i,k}+B_{k,j}\equiv C_{i,j}\pmod Q$ for every $j\in[j_0,j_1]$. Moreover, the segment is active, so $\Ah_{i,k}+\Bh_{k,j}\ne \Ch_{i,j}$ for every $j\in[j_0,j_1]$, and the update contributes exactly the intended triples to $s'_{i,j}$.

For the running time, scanning all segments in $S_0$ and performing the constant-time test and update takes $\cO(|S_0|) = \hO(n^{a+b+\mu}/Q)$. The total time to compute prefix sums is $\cO(n^{a+c})$, since for each fixed $i$ we scan $j=1,\ldots,n^c$ once. Together with the time to compute $S_0$, which is $\hO(n^{a+b+\mu}/Q)$ by \cref{lem:compute-S}, this yields the claimed running time.
\end{proof}

\subsection{Finding a Good Modulus \texorpdfstring{$Q$}{Q}}
\label{sec:find-Q}
We now construct a \emph{good modulus} $Q$, namely an integer $Q$ with $M \le Q \le Mn^{o(1)}$ such that the number of active level-$\ell$ segments is $\hO(n^{a+b+\mu}/Q)$ for every $0 \le \ell \le \lm$.

Let $4 \le R \le n^{o(1)}$ be a parameter to be fixed later, and let $\cP$ be the set of primes in $[R/2,R]$. We take $Q$ to be a product of primes $p_1,p_2,\ldots,p_T$ with $p_t \in \cP$ for all $1 \le t \le T$. Intuitively, each new prime factor
filters out many spurious congruences, so the number of active segments shrinks rapidly with $t$. We choose these primes inductively: at step $t \ge 1$, we select a prime $p_t \in \cP$ and multiply it into the current modulus $Q_{t-1}$, obtaining
\[
    Q_t := \prod_{i=1}^t p_i, \qquad \text{where } Q_0:=1.
\]

For the analysis, define $X_{\ell,t}$ as the number of pairs $((i,k,[j_0,j_1]),s)$, where $(i,k,[j_0,j_1])$ is a level-$\ell$ segment and $s\in[-4\cdot 2^\ell,4\cdot 2^\ell]$, such that
\[
    A_{i,k}+B_{k,j_0}-C_{i,j_0}\equiv s \pmod{Q_t}
    \quad\text{and}\quad
    A_{i,k}+B_{k,j_0}-C_{i,j_0}\ne s.
\]

Our inductive invariant for $Q_t$ is that
\begin{equation}
\label{eq:Xlt}
    X_{\ell,t} = n^{a+b+\mu}\cdot \cO\!\left(\frac{\log^2 n}{R}\right)^t
    \qquad \text{for all } 0 \le \ell \le \lm.
\end{equation}
Later, we choose $R$ so that $X_{\ell,T} = \hO(n^{a+b+\mu}/Q_T)$. Then, by the following lemma, $Q_T$ is a good modulus.

\begin{lemma}
\label{lem:bound-active-seg}
For every $0 \le \ell \le \lm$, we have $|S_{\ell}(Q_T)| \le X_{\ell, T}$.
\end{lemma}
\begin{proof}
Let $(i,k,[j_0,j_1]) \in S_{\ell}(Q_T)$. By definition of an active segment, there exists an integer $s\in[-4\cdot 2^\ell,\,4\cdot 2^\ell]$ such that $A_{i,k}+B_{k,j_0}-C_{i,j_0}\equiv s \pmod{Q_T}$ and $\Ah_{i,k}+\Bh_{k,j_0}\ne \Ch_{i,j_0}$. We claim that $A_{i,k}+B_{k,j_0}-C_{i,j_0}\ne s$. Indeed, since $\Ah_{i,k}+\Bh_{k,j_0} \ne\Ch_{i,j_0}$, \cref{fact:big-difference} implies $|A_{i,k}+B_{k,j_0}-C_{i,j_0}| \ge 7M/10$. On the other hand, $|s|\le 4\cdot 2^\ell \le 4\cdot 2^{\lm} < 4M/10 < 7M/10$, so $A_{i,k}+B_{k,j_0}-C_{i,j_0}\ne s$. Therefore, every active segment in $S_{\ell}(Q_T)$ yields a pair $((i,k,[j_0,j_1]),s)$ counted by $X_{\ell,T}$, and hence $|S_{\ell}(Q_T)| \le X_{\ell,T}$.
\end{proof}

We now explain how to choose $p_t$ so that the invariant~\cref{eq:Xlt} continues to hold. For the base case $t=0$, we set $Q_0:=1$. Since $X_{\ell,0}$ counts level-$\ell$ segment--shift pairs, we trivially have $X_{\ell,0}\le n^{a+b+\mu}$ by \cref{fact:num-segments}, and thus $Q_0$ satisfies the invariant. Now assume~\cref{eq:Xlt} holds for $t-1$ and that $Q_{t-1}$ is fixed. For any prime $p\in\cP$, let $X_{\ell,t}(p)$ denote the value of $X_{\ell,t}$ when $Q_t$ is set to $Q_{t-1}\cdot p$. Our goal is to find a prime $p_t\in\cP$ such that \cref{eq:Xlt} holds with respect to $X_{\ell,t}(p_t)$.

Such a prime exists; applying the following lemma inductively yields the entire sequence $p_1,\ldots,p_T$.

\begin{lemma}
\label{lem:exists-pstar}
There exists $p^*\in\cP$ such that
\begin{equation}
\label{eq:p-star}
    X_{\ell,t}(p^*) = X_{\ell,t-1}\cdot \cO\!\left(\frac{\log^2 n}{R}\right)
    \qquad\text{for all } 0\le \ell \le \lm.
\end{equation}
Moreover, for every $0\le \ell \le \lm$, if $p$ is chosen uniformly at random from $\cP$, then
$$
    \Ex_{p\sim\cP}\!\bigl[X_{\ell,t}(p)\bigr]
    = X_{\ell,t-1}\cdot \cO\!\left(\frac{\log n}{R}\right).
$$
\end{lemma}
\begin{proof}
Fix a level $\ell$. Let $\cX_{\ell,t-1}$ be the set of all pairs $((i,k,[j_0,j_1]),s)$ counted by $X_{\ell,t-1}$; that is, all pairs satisfying
$$
    A_{i,k}+B_{k,j_0}-C_{i,j_0}\equiv s \pmod{Q_{t-1}}
    \qquad\text{and}\qquad
    A_{i,k}+B_{k,j_0}-C_{i,j_0}\ne s,
$$
where $(i,k,[j_0,j_1])$ is a level-$\ell$ segment and $s\in[-4\cdot 2^\ell,\,4\cdot 2^\ell]$. For a prime $p\in\cP$, such a pair is counted in $X_{\ell,t}(p)$ if and only if, additionally, $p$ divides $A_{i,k}+B_{k,j_0}-C_{i,j_0}-s$. Therefore,
$$
    X_{\ell,t}(p)
    =\sum_{((i,k,[j_0,j_1]),s)\in \cX_{\ell,t-1}} \mathbf{1}\!\left[p \mid \bigl(A_{i,k}+B_{k,j_0}-C_{i,j_0}-s\bigr)\right].
$$

Now choose $p$ uniformly at random from $\cP$. Fix any $((i,k,[j_0,j_1]),s)\in \cX_{\ell,t-1}$ and define
$$
    \Delta:=A_{i,k}+B_{k,j_0}-C_{i,j_0}-s.
$$
Note that $\Delta\ne 0$ by the definition of $X_{\ell,t-1}$.
Since the entries of $A,B$ and $C$ are bounded by $\cO(n^\mu)$ and $|s|\le 4\cdot 2^\ell\le 4\cdot 2^{\lm}=\cO(M)=n^{O(1)}$, we have $|\Delta|= n^{O(1)}$ and hence $\log|\Delta|=\cO(\log n)$.
Any prime divisor of $\Delta$ in $[R/2,R]$ is at least $R/2$, so the number of such prime divisors is at most
$$
    \frac{\log|\Delta|}{\log(R/2)}=\cO\!\left(\frac{\log n}{\log(R/2)}\right).
$$
By the prime number theorem, $|\cP|=\Theta(R/\log R)$, and thus
$$
    \Pr_{p\sim\cP}[p\mid \Delta]
    = \frac{\cO(\log n/\log(R/2))}{|\cP|}
    = \cO\!\left(\frac{\log n}{R}\cdot \frac{\log R}{\log(R/2)}\right)
    = \cO\!\left(\frac{\log n}{R}\right),
$$
where the last equality uses
$$
    \frac{\log R}{\log(R/2)}=\frac{\log R}{\log R-\log 2}\le 2
    \qquad \text{for } R \ge 4.
$$
Then, by linearity of expectation,
\begin{align*}
    \Ex_{p\sim\cP}[X_{\ell,t}(p)]
    &=
    \sum_{((i,k,[j_0,j_1]),s)\in \cX_{\ell,t-1}} \Pr_{p\sim\cP}\!\left[p \mid \bigl(A_{i,k}+B_{k,j_0}-C_{i,j_0}-s\bigr)\right] \\
    &\le \sum_{((i,k,[j_0,j_1]),s)\in \cX_{\ell,t-1}} \cO\!\left(\frac{\log n}{R}\right)\\
    &= X_{\ell,t-1}\cdot \cO\!\left(\frac{\log n}{R}\right),
\end{align*}
as desired.

To prove the existence of $p^*$, let $L:=\lm+1$. Recall that $2^{\lm}=\Theta(M)$ and $M= n^{O(1)}$, so $L=O(\log M)=O(\log n)$. From the expectation bound proved above, there exists an absolute constant $C_0>0$ such that for every $0 \le \ell \le \lm$,
$$
    \Ex_{p\sim\cP}\!\bigl[X_{\ell,t}(p)\bigr] \le C_0 \cdot X_{\ell,t-1}\cdot \frac{\log n}{R}.
$$
By Markov's inequality, for any $c>0$,
$$
\Pr_{p\sim\cP}\!\left[
X_{\ell,t}(p) > c\cdot C_0\cdot \frac{\log n}{R}\,X_{\ell,t-1}
\right]
\le \frac{1}{c}.
$$
Set $c:=2L$ and apply a union bound over all $0 \le \ell \le \lm$ to get
$$
\Pr_{p\sim\cP}\!\left[
\exists\,\ell\in\{0,\ldots,\lm\}:\ 
X_{\ell,t}(p) > 2LC_0\cdot \frac{\log n}{R}\,X_{\ell,t-1}
\right]
\le \frac{L}{2L}=\frac{1}{2}<1.
$$
Hence, there exists $p^*\in\cP$ such that for all $0\le \ell\le \lm$,
$$
X_{\ell,t}(p^*)
\le 2LC_0\cdot \frac{\log n}{R}\,X_{\ell,t-1}
= X_{\ell,t-1}\cdot O\!\left(\frac{L\log n}{R}\right)
= X_{\ell,t-1}\cdot O\!\left(\frac{\log^2 n}{R}\right).
$$
\end{proof}

We now discuss how to deterministically select these primes. For this, it is convenient to track the following quantities as well.
\begin{enumerate}
    \item Let $Y_{\ell,t}$ denote the number of pairs $((i,k,[j_0,j_1]),s)$, where $(i,k,[j_0,j_1])$ is a level-$\ell$ segment and $s\in[-4\cdot 2^\ell,\,4\cdot 2^\ell]$, satisfying
    $$
        A_{i,k}+B_{k,j_0}-C_{i,j_0}\equiv s \pmod{Q_t}.
    $$
    \item Let $Z_{\ell}$ denote the number of pairs $((i,k,[j_0,j_1]),s)$, where $(i,k,[j_0,j_1])$ is a level-$\ell$ segment and $s\in[-4\cdot 2^\ell,\,4\cdot 2^\ell]$, satisfying
    $$
        A_{i,k}+B_{k,j_0}-C_{i,j_0}=s.
    $$
\end{enumerate}
As with $X_{\ell,t}(p)$, let $Y_{\ell,t}(p)$ denote the value of $Y_{\ell,t}$ when the modulus is $Q_{t-1}\cdot p$.

The high-level idea for selecting $p_t$ is as follows. Fix a level $\ell$. As briefly mentioned in \cref{sec:overview}, directly minimizing $X_{\ell,t}(p)$ would require naively enumerating all level-$\ell$ segments, which is expensive. The key observation is that
$$
    X_{\ell,t}(p)=Y_{\ell,t}(p)-Z_\ell,
$$
where $Z_\ell$ depends only on the instance and $\ell$ and is independent of $p$. Hence, when comparing primes, this unknown offset cancels. In particular,
$$
    X_{\ell,t}(p)-\min_{p'\in\cP}X_{\ell,t}(p')
    =
    Y_{\ell,t}(p)-\min_{p'\in\cP}Y_{\ell,t}(p').
$$
Therefore, it suffices to compute $Y_{\ell,t}(p)$ for all $p\in\cP$ and choose a prime $p$ such that $Y_{\ell,t}(p)$ is close to its minimum \emph{simultaneously for all levels}. The process of computing $Y_{\ell,t}(p)$ will be essentially the same as how we computed $s_{i,j}$: we reduce modular counting to rectangular matrix multiplication over a polynomial ring. The only difference is that we count \emph{segments}, identified by boundary indicators, rather than all column indices. More formally, we have the following lemma.

\begin{lemma}[Computing $Y_{\ell,t}(p)$]
\label{lem:compute_Ylt}
We can compute $Y_{\ell,t}(p)$ for all primes $p \in \cP$ and all levels $0 \le \ell \le \lm$ in time $\hO(R^2 Q_{t-1} n^{\omega(a,b,c)})$.
\end{lemma}
\begin{proof}
Fix a prime $p\in\cP$ and a level $\ell$. We present an algorithm to compute $Y_{\ell,t}(p)$ for the modulus $Q' = Q_{t-1}\cdot p$ in $\hO(Q'\,n^{\omega(a,b,c)})$ time. Since there are $\lm = \cO(\log M)= \cO(\log n)$ levels, $|\cP|=\Theta(R/\log R)=\tcO(R)$, and $Q'=Q_{t-1}\cdot p \le Q_{t-1}R$, the total running time to compute $Y_{\ell,t}(p)$ for all $p \in \cP$ and all levels is \[
\hO(|\cP|\cdot \lm \cdot (RQ_{t-1})\,n^{\omega(a,b,c)}) = \hO(R^2 Q_{t-1} n^{\omega(a,b,c)}).
\]

Recall that $Y_{\ell,t}(p)$ counts pairs $((i,k,[j_0,j_1]),s)$ where $(i,k,[j_0,j_1])$ is a level-$\ell$ segment, $s\in [-4\cdot 2^\ell,\,4\cdot 2^\ell]$, and $A_{i,k}+B_{k,j_0}-C_{i,j_0}\equiv s \pmod{Q'}$.
A level-$\ell$ segment $(i,k,[j_0,j_1])$ is a maximal interval on which both $\left\lfloor B_{k,j}/2^\ell\right\rfloor$ and $\left\lfloor C_{i,j}/2^\ell\right\rfloor$ are constant. Therefore, $[j_0,j_1]$ starts at column $j$ if and only if either $j=1$, or at least one of $\left\lfloor B_{k,j}/2^\ell\right\rfloor$ and $\left\lfloor C_{i,j}/2^\ell\right\rfloor$ changes at $j$. Hence, it suffices to check each segment exactly once via its start column $j_0$. More concretely, define an $n^b\times n^c$ boundary indicator matrix $I^B$ by
\[
    I^B_{k,1}:=1,\qquad I^B_{k,j}:=\mathbf{1}\!\left[\left\lfloor B_{k,j-1}/2^\ell\right\rfloor\ne \left\lfloor B_{k,j}/2^\ell\right\rfloor\right] \quad (j\ge 2),
\]
and an $n^a\times n^c$ boundary indicator matrix $I^C$ by
\[
    I^C_{i,1}:=1,\qquad
    I^C_{i,j}:=\mathbf{1}\!\left[\left\lfloor C_{i,j-1}/2^\ell\right\rfloor\ne \left\lfloor C_{i,j}/2^\ell\right\rfloor\right] \quad (j\ge 2).
\]
Then it suffices to count, for each $(i,j)\in[n^a]\times[n^c]$, the indices $k\in[n^b]$ with $I^B_{k,j}=1$ or $I^C_{i,j}=1$, grouped by the residue class of $A_{i,k}+B_{k,j}-C_{i,j}\pmod{Q'}$, and then weight each residue class by the number of admissible shifts $s\in[-4\cdot 2^\ell,\,4\cdot 2^\ell]$ in that class.

To do this, we use the polynomial matrix multiplication technique. Work over the ring $\mathcal{R}:=\mathbb{F}[x]/(x^{Q'}-1)$, where $\mathbb{F}$ is a field of sufficiently large characteristic. Construct an $n^a\times n^b$ matrix $A'$ and two $n^b\times n^c$ matrices $B'$ and $B^{\mathrm{bdry}}$ by
\[
    A'_{i,k}(x):=x^{A_{i,k}\bmod{Q'}},\qquad
    B'_{k,j}(x):=x^{B_{k,j}\bmod{Q'}},\qquad
    B^{\mathrm{bdry}}_{k,j}(x):=I^B_{k,j}\cdot x^{B_{k,j}\bmod{Q'}}.
\]
Compute the matrix products $D^{\mathrm{all}}:=A'\cdot B'$ and $D^{\mathrm{bdry}}:=A'\cdot B^{\mathrm{bdry}}$ over $\mathcal{R}$ using fast rectangular matrix multiplication, and finally construct an $n^a\times n^c$ matrix $D$ defined by
\[
    D_{i,j}(x):=
    \begin{cases}
        D^{\mathrm{all}}_{i,j}(x), & \text{if } I^C_{i,j}=1,\\
        D^{\mathrm{bdry}}_{i,j}(x), & \text{if } I^C_{i,j}=0.
    \end{cases}
\]

For each $r \in \{0, \ldots, Q'-1\}$ and $(i,j)$, let $U_{i,j}(r)$ be the coefficient of $x^r$ in $D_{i,j}(x)$.
Then $U_{i,j}(r)$ is the number of indices $k$ such that $A_{i,k}+B_{k,j} \equiv r \pmod {Q'}$ and either $I^B_{k,j}=1$ or $I^C_{i,j}=1$. Next, define
\[
    W(r):=\left|\left\{s\in[-4\cdot 2^\ell,\,4\cdot 2^\ell]:\ s\equiv r \pmod{Q'}\right\}\right|.
\]
Note that the pair $((i,k, [j_0, j_1]),s)$ is counted in $Y_{\ell,t}(p)$ exactly when $A_{i,k}+B_{k,j_0}-C_{i,j_0}\equiv s \pmod{Q'}$. Equivalently, if $r\equiv A_{i,k}+B_{k,j_0}\pmod{Q'}$, then $s\equiv r-C_{i,j_0}\pmod{Q'}$. Thus, grouping by $r$, the contribution of a fixed $(i,j)$ to $Y_{\ell,t}(p)$ is
\[
    \sum_{r=0}^{Q'-1} U_{i,j}(r)\cdot W\!\bigl((r-C_{i,j})\bmod Q'\bigr),
\]
and therefore
\[
    Y_{\ell,t}(p) =\sum_{i\in[n^a]}\sum_{j\in[n^c]}\sum_{r=0}^{Q'-1} U_{i,j}(r)\cdot W\!\bigl((r-C_{i,j})\bmod Q'\bigr).
\]

For the running time, arithmetic in $\mathcal{R}$ reduces exponents modulo $Q'$, so the cost of the rectangular matrix product is $\hO(Q'\,n^{\omega(a,b,c)})$, and we compute $\cO(1)$ such products. Extracting $U_{i,j}(r)$ for all $(i,j)$ and $r$ costs $\cO(Q'\,n^{a+c})$, computing $W(r)$ for all $r$ costs $\cO(Q')$, and evaluating the triple sum costs $\cO(Q'\,n^{a+c})$. Hence, for fixed $p$ and $\ell$ we can compute $Y_{\ell,t}(p)$ in $\hO(Q'\,n^{\omega(a,b,c)})$ time.
\end{proof}

\begin{lemma}
\label{lem:finding-pt}
We can find a prime $p_t\in\cP$ such that $p_t$ satisfies \cref{eq:Xlt}; that is
\[
    X_{\ell,t}(p_t)= n^{a + b + \mu}\cdot \cO\!\left(\frac{\log^2 n}{R}\right)^t \qquad\text{for all } 0\le \ell \le \lm.
\]
Moreover, $p_t$ can be found in $\hO\bigl(R^2 Q_{t-1} n^{\omega(a,b,c)}\bigr)$ time.
\end{lemma}
\begin{proof}
We run the algorithm in
\cref{lem:compute_Ylt} to compute $Y_{\ell,t}(p)$ for all primes $p\in\cP$ and all levels $0\le \ell\le \lm$.
For each $\ell$, define
$$
    Y^*_{\ell,t}:=\min_{p\in\cP} Y_{\ell,t}(p),
$$
and choose $p_t\in\cP$ that minimizes
$$
    \Phi(p):=\max_{0\le \ell \le \lm}\bigl\{Y_{\ell,t}(p)-Y^*_{\ell,t}\bigr\}.
$$

For correctness, we first note that $Y^*_{\ell,t}\ge Z_\ell$.
Since $X_{\ell,t}(p)=Y_{\ell,t}(p)-Z_\ell$ for all $p \in \cP$ and $X_{\ell,t}(p)\ge 0$, we have $Y_{\ell,t}(p)\ge Z_\ell$ for every $p\in\cP$, and thus $Y^*_{\ell,t}\ge Z_\ell$. 

Now, let $p^*\in\cP$ be a prime satisfying \cref{eq:p-star}, which exists by \cref{lem:exists-pstar}. Then, for every $\ell$,
$$
    Y_{\ell,t}(p^*)-Y^*_{\ell,t}
    \le Y_{\ell,t}(p^*)-Z_\ell
    = X_{\ell,t}(p^*)
    = X_{\ell,t-1}\cdot \cO\!\left(\frac{\log^2 n}{R}\right)
    = n^{a+b+\mu}\cdot \cO\!\left(\frac{\log^2 n}{R}\right)^t,
$$
where the last inequality follows from the inductive hypothesis \cref{eq:Xlt}. Therefore,
$$
    \Phi(p^*)= n^{a+b+\mu}\cdot \cO\!\left(\frac{\log^2 n}{R}\right)^t,
$$
and by the choice of $p_t$ minimizing $\Phi$, we obtain
$$
    Y_{\ell,t}(p_t)-Y^*_{\ell,t}
    \le \Phi(p_t)
    \le \Phi(p^*)
    \le n^{a+b+\mu}\cdot \cO\!\left(\frac{\log^2 n}{R}\right)^t \qquad \text{for all } 0 \le \ell \le \lm.
$$

Next, define
$$
    X^*_{\ell,t}:=\min_{p\in\cP} X_{\ell,t}(p).
$$
By minimality and \cref{lem:exists-pstar},
$$
    X^*_{\ell,t}\le \Ex_{p\sim\cP}\!\bigl[X_{\ell,t}(p)\bigr]
    = X_{\ell,t-1}\cdot \cO\!\left(\frac{\log n}{R}\right) = X_{\ell,t-1}\cdot \cO\!\left(\frac{\log^2 n}{R}\right)
$$
Using the inductive hypothesis \cref{eq:Xlt}, we obtain
$$
    X^*_{\ell,t}
    = n^{a+b+\mu}\cdot \cO\!\left(\frac{\log^2 n}{R}\right)^{t-1}\cdot \cO\!\left(\frac{\log^2 n}{R}\right)
    = n^{a+b+\mu}\cdot \cO\!\left(\frac{\log^2 n}{R}\right)^t,
$$
Also, since $X_{\ell,t}(p)=Y_{\ell,t}(p)-Z_\ell$ for all $p\in\cP$ and $Z_\ell$ is independent of $p$, we have
$$
    Y^*_{\ell,t}-Z_\ell
    =\min_{p\in\cP}\bigl\{Y_{\ell,t}(p)-Z_\ell\bigr\}
    =\min_{p\in\cP} X_{\ell,t}(p)
    =X^*_{\ell,t}.
$$
Combining the above bounds, for every $0\le \ell \le \lm$ we get
$$
    X_{\ell,t}(p_t) = Y_{\ell,t}(p_t)-Z_\ell = \bigl(Y^*_{\ell,t}-Z_\ell\bigr) + \bigl(Y_{\ell,t}(p_t)-Y^*_{\ell,t}\bigr) \le X^*_{\ell,t} + \Phi(p_t) = n^{a+b+\mu}\cdot \cO\!\left(\frac{\log^2 n}{R}\right)^t.
$$

For the running time, computing $Y_{\ell,t}(p)$ for all $p\in\cP$ and all $0\le \ell\le \lm$ via \cref{lem:compute_Ylt} dominates, giving $\hO\bigl(R^2 Q_{t-1} n^{\omega(a,b,c)}\bigr)$ time. The remaining steps (taking minima and selecting $p_t$) take at most $\cO(|\cP|\lm)=\tcO(R\log n)$ additional time, which is lower order.
\end{proof}

It remains to set the parameters so that the above procedure yields a good modulus. To this end, we prove the following fact.

\begin{fact}
\label{fact:subpoly}
Let $R=2^{\Theta\!\left(\sqrt{\log n/\log\log n}\right)}$ and $T = \cO(\log n/\log R)$. Then $(\log n)^{2T} = n^{o(1)}$.
\end{fact}
\begin{proof}
We write $(\log n)^{2T} = \exp\!\bigl(2T\log\log n\bigr)$. Thus, it suffices to show that $\exp(2T\log\log n)=n^{o(1)}$. Indeed,
\[
    \exp\!\bigl(2T\log\log n\bigr) = \exp\!\left(2\cdot \cO\!\left(\frac{\log n}{\log R}\right)\cdot \log\log n\right) = \exp\!\left(\cO\!\left(\log n\cdot \frac{\log\log n}{\log R}\right)\right).
\]
By our choice of $R$, we have $\log R = \Theta(\sqrt{\log n/ \log \log n})$, which implies $\log \log n/ \log R = o(1)$. Therefore,
$$
    \exp\!\bigl(2T\log\log n\bigr)=\exp\!\bigl(o(\log n)\bigr)=n^{o(1)},
$$
as claimed.
\end{proof}

\begin{lemma}[Computing a good modulus]
\label{lem:compute-Q}
We can compute a good modulus $Q$ in time $\hO(Mn^{\omega(a,b,c)})$.
\end{lemma}
\begin{proof}
By applying \cref{lem:finding-pt} inductively for $T$ steps, we obtain primes $p_1,\ldots,p_T\in\cP$ and a modulus $Q_T:=\prod_{t=1}^T p_t$ such that \cref{eq:Xlt} holds. Since $p_t\in[R/2,R]$ for all $1 \le t \le T$, we have $(R/2)^T\le Q_T\le R^T$. In particular, to ensure $Q_T\ge M$ it suffices to take 
\[
    T=\Theta(\log M / \log R)=\cO(\log n / \log R).
\]
Moreover, if we take $T$ to be the \emph{first} index such that $Q_T\ge M$, then $Q_{T-1}<M$ and hence
$Q_T=Q_{T-1}\cdot p_T < M\cdot R$, so $Q_T$ overshoots $M$ by at most a factor of $R$. We set
\[
    R=2^{\Theta\!\left(\sqrt{\log n/\log\log n}\right)}.
\]
Then $R=n^{o(1)}$, and therefore $M\le Q_T\le M\cdot R = M\cdot n^{o(1)}$. Finally, by \cref{fact:subpoly}, we have $(\log^2 n)^T=n^{o(1)}$, and so 
\[
    (\log^2 n/R)^T = (\log n)^{2T}/R^T = n^{o(1)}/R^T = n^{o(1)}/Q_T.
\]
Hence $Q_T$ satisfies $M\le Q_T\le Mn^{o(1)}$, and
\[
    X_{\ell,T}
    = n^{a+b+\mu}\cdot \cO\!\left(\frac{\log^2 n}{R}\right)^T
    = \frac{n^{a+b+\mu+o(1)}}{Q_T}
    \qquad\text{for all } 0\le \ell\le \lm.
\]
Therefore, $Q_T$ is a good modulus.

For the running time, we apply \cref{lem:finding-pt} for $T$ steps. At step $t$, the cost of computing $p_t$ is $\hO\!\left(R^2 Q_{t-1} n^{\omega(a,b,c)}\right)$. Since $Q_{t-1}\le Q_{T-1}<M$, the total time is
\[
    \hO\!\left(\sum_{t=1}^T R^2 Q_{t-1} n^{\omega(a,b,c)}\right)
    = \hO\!\left(T R^2 M\, n^{\omega(a,b,c)}\right)
    = \hO\!\left(M n^{\omega(a,b,c)}\right),
\]
where the last bound uses $TR^2=n^{o(1)}$ for our choice of $R$.
\end{proof}

\subsection{Combining the Results}
\begin{theorem}
\label{thm:exact-tri-mod-M}
There is a deterministic $\hO(Mn^{\omega(a,b,c)}+n^{a+b+\mu}/M)$-time
algorithm for \cref{prob:exact-tri-mod-M}.
\end{theorem}

\begin{proof}
We first compute a good modulus $Q$ via \cref{lem:compute-Q} in time
$\hO(Mn^{\omega(a,b,c)})$. Next, we compute $s_{i,j}$ for all $(i,j)\in[n^a]\times[n^c]$ via \cref{lem:compute-s} in time $\hO(Qn^{\omega(a,b,c)}) = \hO(Mn^{\omega(a,b,c)})$. Finally, we compute $s'_{i,j}$ for all $(i,j)\in[n^a]\times[n^c]$ via \cref{lem:compute-sprime} in time $\hO\!\left(n^{a+b+\mu}/Q+n^{a+c}\right) = \hO\!\left(n^{a+b+\mu}/M+n^{a+c}\right)$. For each $(i,j)\in[n^a]\times[n^c]$, we output \textsc{Yes} if and only if $s_{i,j} > s'_{i,j}$; correctness follows from \cref{lem:s-vs-sprime}. The total running time is $\hO(Mn^{\omega(a,b,c)} + n^{a+b+\mu}/M)$ as claimed.
\end{proof}

\begin{proof}[Proof of \cref{thm:main}]
Our claimed running time is
\[
    \hO\left(n^{a+c}+n^{b+c}+n^{(a+b+\mu+\omega(a,b,c))/2}\right).
\]
Let $d := (a+b+\mu-\omega(a,b,c))/2$. Since $\mu\ge 0$, we have $d\le \mu$: indeed, $d>\mu$ would imply $\mu+\omega(a,b,c)<a+b$, contradicting the trivial lower bound $\omega(a,b,c)\ge a+b$. Furthermore, we may assume that $d\ge 0$: if $d<0$, then $a+b+\mu<\omega(a,b,c)$. In this case, Lemma~17 of~\cite{DBLP:conf/icalp/Gu0WX21} gives an algorithm running in time $\hO\left(n^{a+c}+n^{b+c}+n^{a+b+\mu}\right)$, which is dominated by our claimed bound. Hence, we may assume $0\le d\le \mu$.

Now, by \cref{lem:reduction-to-verification,lem:reduction-to-small-modM}, it suffices to solve \cref{prob:exact-tri-mod-M}. By \cref{thm:exact-tri-mod-M}, there is a deterministic algorithm for \cref{prob:exact-tri-mod-M} that runs in time $\hO(Mn^{\omega(a,b,c)}+n^{a+b+\mu}/M)$. Take $d$ as the exponent for $M$; that is, take $M=\Theta\left(n^{(a+b+\mu-\omega(a,b,c))/2}\right)$. Then the two terms balance, yielding $\hO\left(Mn^{\omega(a,b,c)}+n^{a+b+\mu}/M\right)=\hO\left(n^{(a+b+\mu+\omega(a,b,c))/2}\right)$, which is dominated by our claimed bound.
\end{proof}

%% file: column.tex
\section{Extension to Column-Monotone Min-Plus Product}
\label{sec:column}

In this section, we present a deterministic algorithm for column-monotone Min-Plus product, whose definition we recall below.

\begin{prob}
\label{prob:column-monotone-min-plus}
    Let $A$ be an $n^a \times n^b$ integer matrix. Let $B$ be an $n^b \times n^c$ column-monotone integer matrix with entries in $[n^\mu]$. Given $(A,B)$ as input, compute $A \star B$.
\end{prob}

We start by reducing column-monotone Min-Plus product to the following problem, which is a ``rotated'' version of \cref{prob:exact-tri} (the differences have been highlighted).

\begin{customprob}{\ref*{prob:exact-tri}'}
\label{prob:exact-tri-prime}
    Let $A$ be an $n^a \times n^b$ integer matrix with nonnegative entries bounded by $\cO(n^{\mu})$. Let $B$ be an $n^b \times n^c$ row-monotone integer matrix with nonnegative entries bounded by $\cO(n^{\mu})$. Let $C$ be an $n^a \times n^c$ row-monotone integer matrix with nonnegative entries bounded by $\cO(n^{\mu})$. Given $(A,B,C)$ as input, \textbf{for every $\boldsymbol{(i,k) \in [n^{a}] \times [n^{b}]}$, decide whether there exists $\boldsymbol{j \in [n^{c}]}$} such that $A_{i,k} + B_{k,j} = C_{i,j}$.
\end{customprob}

We will show that the running time for \cref{prob:exact-tri-prime} is the same as the running time for \cref{prob:exact-tri}.

\begin{theorem}
\label{thm:alg-prob:exact-tri-prime}
    There is a deterministic algorithm for \cref{prob:exact-tri-prime} that runs in time $\hO(n^{a+c} + n^{b+c} + n^{(a+b+\mu+\omega(a,b,c))/2})$.
\end{theorem}

\begin{proof}[Proof of \cref{thm:main-column} assuming \cref{thm:alg-prob:exact-tri-prime}]
    Let $(A, B)$ be an instance of \cref{prob:column-monotone-min-plus}. Similar to the proof of \cref{lem:reduction-to-verification}, we may assume that all entries of $A$ are nonnegative integers bounded by $\cO(n^\mu)$. By an observation in~\cite{CDXZ22}, we may also assume that each row of $A$ is monotonically non-increasing: if there exist $i,k_1,k_2$ such that $k_1<k_2$ and $A_{i,k_1}<A_{i,k_2}$, then $A_{i,k_1}+B_{k_1,j}<A_{i,k_2}+B_{k_2,j}$ for every $j$ by the monotonicity of $B$. Hence, replacing $A_{i,k_2}$ by $A_{i,k_1}$ does not change $A\star B$.
    
    Now we follow the same recursive halving approach as in the proof of \cref{lem:reduction-to-verification}. Let $A'_{i,j}:=\lfloor A_{i,j}/2\rfloor$ and $B'_{i,j}:=\lfloor B_{i,j}/2\rfloor$. Flooring preserves monotonicity, so the columns of $B'$ are monotone and the rows of $A'$ are monotonically non-increasing. We compute $C':=A'\star B'$ recursively, with the trivial base case when both matrices are zero. Define three $n^a\times n^c$ candidate matrices by $C^{(s)}_{i,j}:=2C'_{i,j}+s$ for $s\in\{0,1,2\}$.
    
    As in the proof of \cref{lem:reduction-to-verification}, it suffices to check, for each $s\in\{0,1,2\}$ and $(i,j)\in[n^a]\times[n^c]$, whether there exists $k\in[n^b]$ such that $A_{i,k}+B_{k,j}=C^{(s)}_{i,j}$. Let $W=\cO(n^\mu)$ be the maximum entry of $A$, $B$, and $C^{(s)}$. This condition is equivalent to $(W-C^{(s)}_{i,j})+B_{k,j}=W-A_{i,k}$. Thus, it can be solved by \cref{prob:exact-tri-prime} on the input $(W^{n^a\times n^c}-C^{(s)}, B^T, W^{n^a\times n^b}-A)$, where $W^{m\times p}$ denotes the $m\times p$ all-$W$ matrix. All entries are bounded by $\cO(n^\mu)$, and both $B^T$ and $W^{n^a\times n^b}-A$ are row-monotone, so the input satisfies the requirements of \cref{prob:exact-tri-prime}. The dimensions are rotated: the three matrices have dimensions $n^a\times n^c$, $n^c\times n^b$, and $n^a\times n^b$. Hence, by \cref{thm:alg-prob:exact-tri-prime} and the symmetry $\omega(a,c,b)=\omega(a,b,c)$, each verifier call takes time $\hO\!\left(n^{a+b}+n^{b+c}+n^{(a+c+\mu+\omega(a,b,c))/2}\right)$. Since there are only constantly many verifier calls per recursion level and the recursion depth is $\cO(\log n)$, the final asymptotic running time remains the same.
\end{proof}

\subsection{Algorithm for \texorpdfstring{\cref{prob:exact-tri-prime}}{{Problem \ref*{prob:exact-tri-prime}}}}

Using the same proof as for \cref{lem:reduction-to-small-modM}, we can reduce \cref{prob:exact-tri-prime} to the following problem. 

\begin{customprob}{\ref*{prob:exact-tri-mod-M}'}
\label{prob:exact-tri-mod-M-prime}
    Let $(A,B,C)$ be an instance of \cref{prob:exact-tri-prime}, and let $M$ be a positive integer that is a multiple of $100$ and satisfies $M = \cO(n^{\mu})$. Additionally, assume that for every $i \in [n^a]$, $k \in [n^b]$, and $j \in [n^c]$, we have
    \[
        A_{i,k} \bmod M,\; B_{k,j} \bmod M,\; C_{i,j} \bmod M \le \frac{M}{10}.
    \]
    Given $(A,B,C)$ as input, for every $(i,k) \in [n^{a}] \times [n^{b}]$, decide whether there exists $j \in [n^{c}]$ such that $A_{i,k} + B_{k,j} = C_{i,j}$.
\end{customprob}

As before, we define
\[
    \Ah_{i,k}=\left\lfloor \frac{A_{i,k}}{M}\right\rfloor,
    \qquad
    \Bh_{k,j}=\left\lfloor \frac{B_{k,j}}{M}\right\rfloor,
    \qquad
    \Ch_{i,j}=\left\lfloor \frac{C_{i,j}}{M}\right\rfloor
\]
and 
\[
    \Al_{i,k}=A_{i,k}\bmod M,
    \qquad
    \Bl_{k,j}=B_{k,j}\bmod M,
    \qquad
    \Cl_{i,j}=C_{i,j}\bmod M.
\]

Our algorithm proceeds as follows:
\begin{enumerate}
    \item Find a ``good'' modulus $Q$ with respect to \cref{def:good-modulus}. By \cref{lem:compute-Q}, it takes $\hO(M n^{\omega(a, b, c)})$ time. 
    \item For every $(i,k)\in[n^a]\times[n^b]$, compute $r_{i,k}$, the number of $j\in[n^c]$ such that
    $A_{i,k}+B_{k,j}\equiv C_{i,j}\pmod{Q}$. Using the same method as \cref{lem:compute-s}, it takes $\hO(Q n^{\omega(a, c, b)}) = \hO(M n^{\omega(a, b, c)})$ time. 
    \item For every $(i,k)\in[n^a]\times[n^b]$, compute $r'_{i,k}$, the number of $j\in[n^c]$ such that
    $A_{i,k}+B_{k,j}\equiv C_{i,j}\pmod{Q}$ and $\Ah_{i,k}+\Bh_{k,j}\ne \Ch_{i,j}$. To do so, first compute $S_0(Q)$ using \cref{lem:compute-S} in $\hO(n^{a+b+\mu}/Q) = \hO(n^{a+b+\mu} / M)$ time. Fix any triple $(i,k,j)$ that is counted in $r'_{i,k}$, i.e., satisfying $A_{i,k}+B_{k,j}\equiv C_{i,j}\pmod Q$ and $\Ah_{i,k}+\Bh_{k,j}\ne \Ch_{i,j}$. Let $(i,k,[j_0,j_1])$ be the unique level-$0$ segment such that $j \in [j_0, j_1]$. Since $(i,k,[j_0,j_1])$ is a level-$0$ segment, the values $B_{k,j}$ and $C_{i,j}$ are constant over all $j\in[j_0,j_1]$. We have $j \in [j_0, j_1]$, so 
    \[
        \Ah_{i,k}+\Bh_{k,j_0} = \Ah_{i,k}+\Bh_{k,j} \ne \Ch_{i,j} = \Ch_{i,j_0}
    \]
    and 
    \[
        A_{i,k}+B_{k,j_0} \equiv A_{i,k}+B_{k,j} \equiv C_{i,j}  \equiv C_{i,j_0}\pmod Q.
    \]
    Hence, $(i,k,[j_0,j_1])$ is an active level-$0$ segment in $S_0(Q)$. Therefore, we can enumerate all segments $(i, k, [j_0, j_1])$ in $S_0(Q)$, check whether $A_{i, k} + B_{k, j_0} \equiv C_{i, j_0} \pmod{Q}$, and if so, increment $r'_{i, k}$ by $j_1 - j_0 + 1$. Since $Q$ is a good modulus, we have $|S_0(Q)| =\hO(n^{a+b+\mu} / M)$, so this step takes $\hO(n^{a+b+\mu} / M)$ time.
    \item For each $(i,k)\in[n^a]\times[n^b]$, output \textsc{Yes} if and only if $r_{i,k}>r'_{i,k}$.
\end{enumerate}

The correctness of the algorithm follows from the following lemma. 

\begin{lemma}
\label{lem:r-vs-rprime}
    For every $(i,k)\in[n^a]\times[n^b]$, there exists $j\in[n^c]$ such that $A_{i,k}+B_{k,j}=C_{i,j}$ if and only if $r_{i,k}>r'_{i,k}$.
\end{lemma}
\begin{proof}
    Fix $(i,k)$ and suppose there exists $j\in[n^c]$ such that $A_{i,k}+B_{k,j}=C_{i,j}$. Then $A_{i,k}+B_{k,j}\equiv C_{i,j}\pmod{Q}$, so this $j$ contributes one to $r_{i,k}$. Moreover, under the promise of \cref{prob:exact-tri-mod-M-prime}, we have $\Al_{i,k}+\Bl_{k,j}\le 2M/10$, so adding the low parts creates no carry across multiples of $M$. Hence, $\Ch_{i,j} =\lfloor C_{i,j}/M \rfloor =\lfloor (A_{i,k}+B_{k,j})/M\rfloor =\Ah_{i,k}+\Bh_{k,j}$. Thus, this $j$ does not contribute to $r'_{i,k}$, and therefore $r_{i,k}>r'_{i,k}$.
    
    Conversely, suppose $r_{i,k}>r'_{i,k}$.
    Then there exists $j\in[n^c]$ such that $A_{i,k}+B_{k,j}\equiv C_{i,j}\pmod{Q}$ and $ \Ah_{i,k}+\Bh_{k,j}=\Ch_{i,j}$. By \cref{fact:small-difference}, we have $|A_{i,k}+B_{k,j}-C_{i,j}|\le 3M/10 < Q$. Since $A_{i,k}+B_{k,j}-C_{i,j}$ is a multiple of $Q$ whose magnitude is strictly less than $Q$, it must be zero. Hence, $A_{i,k}+B_{k,j}=C_{i,j}$.
\end{proof}

Now the claimed running time in \cref{thm:alg-prob:exact-tri-prime} is
\[
    \hO\left(n^{a+c}+n^{b+c}+n^{(a+b+\mu+\omega(a,b,c))/2}\right).
\]
Let $d:=(a+b+\mu-\omega(a,b,c))/2$. As discussed in the proof of \cref{thm:main}, we have $d\le \mu$.

If $d\ge 0$, our approach runs in time $\hO(Mn^{\omega(a,b,c)}+n^{a+b+\mu}/M)$. Taking $M=\Theta(n^{(a+b+\mu-\omega(a,b,c))/2})$ balances the two terms and gives a running time of $\hO(n^{(a+b+\mu+\omega(a,b,c))/2})$, which is dominated by the claimed bound.

If $d<0$, then $a+b+\mu<\omega(a,b,c)$. In this case, we use a direct approach. For every row of $B$ and $C$, compute its maximal constant intervals by scanning the row; that is, we partition each row into maximal contiguous intervals on which the row value is constant. This takes $\cO(n^{b+c}+n^{a+c})$ time. Since the entries are bounded by $\cO(n^\mu)$ and the rows are monotone, each row has $\cO(n^\mu)$ such intervals. Then, for each pair $(i,k)\in[n^a]\times[n^b]$, scan the two interval decompositions of $B_{k,*}$ and $C_{i,*}$ using two pointers, and consider their common refinement: the maximal intervals obtained by intersecting one constant interval of $B_{k,*}$ with one constant interval of $C_{i,*}$. On each such refined interval, both $B_{k,j}$ and $C_{i,j}$ are constant, so it suffices to check whether $A_{i,k}+B_{k,j}=C_{i,j}$ for one representative $j$ in the interval. Since the common refinement has $\cO(n^\mu)$ intervals, all pairs $(i,k)$ can be processed in $\cO(n^{a+b+\mu})$ time. Thus the total running time is $\tcO(n^{a+c}+n^{b+c}+n^{a+b+\mu})$, and since $a+b+\mu<\omega(a,b,c)$, this running time is dominated by the claimed bound.

%% file: convolution.tex
\section{Deterministic Monotone Min-Plus Convolution}
\label{sec:convolution}

In this section, we present a deterministic algorithm for monotone Min-Plus convolution, whose definition we recall below.

\begin{prob}
\label{prob:monotone-min-plus-conv}
    Let $A$ and $B$ be two monotone arrays of length $n$ with entries in $[n^\mu]$. Given $(A,B)$ as input, compute $A \dia B$.
\end{prob}

\subsection{Reduction to a Verification Problem}
\label{subsec:reductions-conv}
As in \cref{sec:product}, we first reduce \cref{prob:monotone-min-plus-conv} to a corresponding verification problem and then to a promised version in which all entries have small residues modulo a parameter $M$.

\begin{prob}
\label{prob:exact-tri-conv}
    Let $A$ and $B$ be two monotone arrays of length $n$ with nonnegative entries bounded by $\cO(n^{\mu})$. Let $C$ be an array of length-$2n-1$ indexed by $2$ to $2n$ with nonnegative entries bounded by $\cO(n^{\mu})$. Given $(A,B,C)$ as input, for every $k \in [2,2n]$, decide whether there exists $i \in [n]$ such that $k - i \in [n]$ and $A_i + B_{k-i} = C_k$.
\end{prob}

\begin{prob}
\label{prob:exact-tri-mod-M-conv}
    Let $A$, $B$, and $C$ be an instance of \cref{prob:exact-tri-conv}, and let $M$ be a positive integer that is a multiple of $100$ and at most $\cO(n^{\mu})$. Additionally, assume that
    \[
        A_i \bmod M, B_i \bmod M, C_k \bmod M \le \frac{M}{10} \qquad \text{ for } \qquad i \in [n], k \in [2, 2n].
    \]
    Given $(A,B,C)$ as input, for every $k \in [2,2n]$, decide whether there exists $i \in [n]$ such that $k - i \in [n]$ and $A_i + B_{k-i} = C_k$.
\end{prob}

\begin{lemma}
\label{lem:reduction-to-verification-conv}
    If there is a $T(n; n^{\mu}, M)$-time deterministic algorithm for \cref{prob:exact-tri-mod-M-conv} for any choice of $M$, then there is an $\tcO(T(n; n^{\mu}, M))$-time deterministic algorithm for \cref{prob:monotone-min-plus-conv}.
\end{lemma}
\begin{proof}[Proof sketch]
    The proof is similar to the proofs of \cref{lem:reduction-to-verification} and \cref{lem:reduction-to-small-modM}. We briefly describe the constructions.
    
    To reduce \cref{prob:monotone-min-plus-conv} to \cref{prob:exact-tri-conv}, we let $A'_i := \lfloor A_i/2 \rfloor$ and $B'_i := \lfloor B_i/2 \rfloor$. We compute $C' := A' \dia B'$ recursively. Construct three candidate arrays by $ C^{(s)}_k := 2C'_k + s$ for all $s \in \{0,1,2\}$. Note that flooring preserves monotonicity, so $A'$ and $B'$ are also monotone. Therefore, we run the verifier for \cref{prob:exact-tri-conv} on $(A,B,C^{(s)})$ for each $s \in \{0,1,2\}$. By an argument analogous to \cref{lem:reduction-to-verification}, correctness follows.
    
    Next, we reduce \cref{prob:exact-tri-conv} to \cref{prob:exact-tri-mod-M-conv}. Let $W := M/100$. We may assume that all entries of $A$, $B$, and $C$ are at least $M$ by adding $M$ to the entries of $A$ and $B$ and adding $2M$ to the entries of $C$. We partition $\{0,1,\ldots,M-1\}$ into intervals $I_s$ defined by
    \[
        I_s := [sW,(s+1)W) \qquad \text{for } s \in \{0,1,\ldots,99\}.
    \]
    For each $s \in \{0,1,\ldots,99\}$, define arrays $A^{(s)}[1,2, \ldots, n]$ and $B^{(s)}[1,2, \ldots, n]$ by
    \[
        A^{(s)}_{i} := 
        \begin{cases} A_{i} - sW, & \text{if } A_{i} \bmod M \in I_s,\\
     \Big\lfloor \dfrac{A_{i}-sW}{M}\Big\rfloor \cdot M + 3W, & \text{otherwise}
        \end{cases},
    \]
    and define $B^{(s)}$ analogously. Additionally, for all $s,t \in \{0,1,\ldots,99\}$, let
    \[
    J_{s,t} := \bigl\{(s+t)W \bmod M,\ ((s+t)W+1) \bmod M,\ \ldots,\ ((s+t)W + (2W-1))\bmod M\bigr\},
    \]
    and define an array $ C^{(s,t)}[2,3, \ldots, 2n]$ by
    \[
        C^{(s,t)}_{k} := \begin{cases}
                            C_{k} - (s + t)W, & \text{if } C_{k} \bmod M \in J_{s,t},\\
                            \Big\lfloor \dfrac{C_{k}-(s+t)W}{M}\Big\rfloor \cdot M + 7W, & \text{otherwise.}
                        \end{cases}
    \]
    For all $s,t \in \{0,1,\ldots,99\}$, run the algorithm for \cref{prob:exact-tri-mod-M-conv} on the instances $(A^{(s)},B^{(t)},C^{(s,t)})$ and output \textsc{Yes} for index $k$ if any instance outputs \textsc{Yes} for $k$. By an argument analogous to \cref{lem:reduction-to-small-modM}, correctness follows.
    
    For the running time, we make $\cO(1)$ calls to the $T(n; n^{\mu},M)$-time verifier for \cref{prob:exact-tri-mod-M-conv} per recursion level. The recursion depth is $\cO(\log(n^\mu))=\cO(\log n)$, which preserves the $\tcO(T(n; n^{\mu},M))$ bound.
\end{proof}

\subsection{Algorithm for \texorpdfstring{\cref{prob:exact-tri-mod-M-conv}}{Problem \ref*{prob:exact-tri-mod-M-conv}}}
We focus on solving \cref{prob:exact-tri-mod-M-conv} efficiently. Suppose we are given an instance $(A,B,C)$ of \cref{prob:exact-tri-mod-M-conv} with $M = \cO(n^{d})$ for some constant $0 \le d \le \mu$. Define
\[
\Ah_{i}=\left\lfloor \frac{A_{i}}{M}\right\rfloor,
\qquad
\Bh_{j}=\left\lfloor \frac{B_{j}}{M}\right\rfloor,
\qquad
\Ch_{k}=\left\lfloor \frac{C_{k}}{M}\right\rfloor.
\]
We also define
\[
    \Al_{i}=A_{i}\bmod M,
    \qquad
    \Bl_{j}=B_{j}\bmod M,
    \qquad
    \Cl_{k}=C_{k}\bmod M.
\]
The high-level outline of our algorithm is as follows:
\begin{enumerate}
    \item Find a ``good'' modulus $Q$ with $M \le Q \le Mn^{o(1)}$. The definition of ``good'' is deferred to \cref{def:good-modulus-conv}. 
    \item For every $k \in [2,2n]$, compute $s_k$, the number of indices $i \in [n]$ such that $k - i \in [n]$ and
    $A_i + B_{k-i}\equiv C_k \pmod{Q}$.
    \item For every $k \in [2,2n]$, compute $s'_k$, the number of indices $i \in [n]$ such that $k - i \in [n]$ and
    $A_i + B_{k-i}\equiv C_k \pmod{Q}$ and $\Ah_i+\Bh_{k-i}\ne \Ch_k$.
    \item For each $k \in [2,2n]$, output \textsc{Yes} if and only if $s_k>s'_k$.
\end{enumerate}

We will show that $A_i + B_{k-i} = C_k$ holds if and only if $s_k > s'_k$. First, we state two facts without proofs, as they are essentially identical to the proofs of \cref{fact:big-difference} and \cref{fact:small-difference}.

\begin{fact}
\label{fact:big-diff-conv}
    For every $k\in[2,2n]$ and $i\in[n]$ such that $k - i \in [n]$, if $\Ah_i+\Bh_{k-i}\ne \Ch_k$, then $|A_i+B_{k-i}-C_k|\ge 7M/10$.
\end{fact}

\begin{fact}
\label{fact:small-diff-conv}
    For every $k\in[2,2n]$ and $i\in[n]$ such that $k - i \in [n]$, if $\Ah_i+\Bh_{k-i}=\Ch_k$, then $|A_i+B_{k-i}-C_k|\le 3M/10$.
\end{fact}

\begin{lemma}
\label{lem:s-vs-sprime-conv}
    For every $k \in [2,2n]$, there exists $i \in [n]$ such that $k - i \in[n]$ and $A_i+B_{k-i}=C_k$ if and only if $s_k>s'_k$.
\end{lemma}
\begin{proof} 
    Suppose there exists $i\in[n]$ such that $k-i\in[n]$ and $A_i + B_{k-i} = C_k$. Then $A_i + B_{k-i} \equiv C_k \pmod{Q}$, and $\Ah_i + \Bh_{k-i} = \Ch_k$ because adding the low parts creates no carry across a multiple of $M$. Hence, $i$ contributes to $s_k$ but not to $s'_k$. Conversely, suppose $s_k > s'_k$. Then there exists $i\in[n]$ such that $A_i + B_{k-i} \equiv C_k \pmod Q$ and $\Ah_i + \Bh_{k-i} = \Ch_k$. Then $A_i + B_{k-i} - C_k$ is a multiple of $Q$, and $|A_i + B_{k-i} - C_k| < Q$ by \cref{fact:small-diff-conv}. This implies $A_i + B_{k-i} = C_k$.
\end{proof}

Let $\lm$ be such that $M/20 \le 2^{\lm} < M/10$. We define segments, active segments, and good modulus analogously to \cref{sec:product}.

\begin{definition}[Segments] 
\label{def:segments-conv} 
    For $0 \le \ell \le \lm$, a level-$\ell$ segment is a pair $([i_0, i_1], k)$ such that $[i_0,i_1]\subseteq [\max\{1,k-n\},\min\{n,k-1\}]$, for every $ i \in [i_0, i_1]$, $\lfloor A_{i_0} / 2^\ell\rfloor = \lfloor A_{i} / 2^\ell\rfloor$, and $\lfloor B_{k -i_0} / 2^\ell\rfloor = \lfloor B_{k -i} / 2^\ell\rfloor$, and $[i_0, i_1]$ cannot be extended further. 
\end{definition}

\begin{definition}[Active Segment]
\label{def:active-segment-conv}
    For $0 \le \ell \le \lm$, a level-$\ell$ segment $([i_0,i_1],k)$ is called \emph{active} with respect to $Q$ if $\Ah_{i_0}+\Bh_{k-i_0}-\Ch_k \ne 0$ and there exists $s\in[-4\cdot 2^\ell,\,4\cdot 2^\ell]$ such that $A_{i_0}+B_{k-i_0}-C_k \equiv s \pmod Q$.
\end{definition}

\begin{fact}
\label{fact:num-segments-conv}
    For $0 \le \ell \le \lm$, the number of level-$\ell$ segments is $\cO(n^{1+\mu}/2^\ell)$.
\end{fact}

\begin{definition}[Set of Active Segments]
\label{def:active-segments-conv}
    For $0 \le \ell \le \lm$, let $S_{\ell}(Q)$ be the set of all active level-$\ell$ segments with respect to $Q$.
\end{definition}

\begin{definition}[Good Modulus]
\label{def:good-modulus-conv}
    We say that an integer $Q$ is a \emph{good modulus} if $M \le Q \le Mn^{o(1)}$ and $|S_{\ell}(Q)| \le \hO(n^{1+\mu}/Q)$ for all $0\le \ell \le \lm$.
\end{definition}

\subsection{Computing \texorpdfstring{$s_k$}{sk} and \texorpdfstring{$s_k'$}{sk'}}
Recall that $s_k$ is the number of indices $i \in [n]$ such that $k - i \in [n]$ and $A_i + B_{k-i} \equiv C_k \pmod{Q}$. We compute $s_k$ for all $k \in [2,2n]$ using polynomial multiplication.

\begin{lemma}[Computing $s_k$]
\label{lem:compute-s-conv}
    We can compute $s_k$ for every $k \in [2,2n]$ in $\tcO(Qn)$ time.
\end{lemma}
\begin{proof}
    Work over the ring $\mathcal{R} := \mathbb{F}[x]/(x^Q-1)$, where $\mathbb{F}$ has characteristic greater than $n$. Define two bivariate polynomials in $\mathcal{R}[y]$ by
    \[
        P_A(x,y) := \sum_{i=1}^n x^{A_i \bmod Q} y^i,
        \qquad
        P_B(x,y) := \sum_{j=1}^n x^{B_j \bmod Q} y^j.
    \]
    Compute the product $P(x,y) := P_A(x,y)\cdot P_B(x,y)$ in $\mathcal{R} [y]$.
    
    For correctness, fix $k\in[2,2n]$ and $r\in\{0,1,\ldots,Q-1\}$. The coefficient of $x^r y^k$ in $P(x,y)$ equals the number of pairs $(i,j)\in[n]\times[n]$ such that $i+j=k$ and $A_i+B_j\equiv r \pmod Q$. Equivalently, it equals the number of indices $i\in[n]$ such that $k-i\in [n]$ and $A_i+B_{k-i}\equiv r \pmod Q$. Hence, for each $k\in[2,2n]$, we set $s_k$ to be the coefficient of $x^{C_k \bmod Q}y^k$ in $P(x,y)$.
    
    For the running time, the degrees in $x$ and $y$ of $P_A$ and $P_B$ are at most $Q -1$ and $n$. Therefore, using fast polynomial multiplication, we can compute $P$ in $\tcO(Qn)$ time.
\end{proof}

We now describe how to compute $s'_k$. Recall that $s'_k$ is the number of indices $i\in[n]$ such that $k - i \in [n]$, $A_i+B_{k-i}\equiv C_k \pmod Q$ and $\Ah_i+\Bh_{k-i}\ne \Ch_k$. Assume that $Q$ is a good modulus, as in \cref{def:good-modulus-conv}. For simplicity, let $S_{\ell}$ denote the set of all active level-$\ell$ segments, suppressing the dependence on $Q$. We first prove a property that we will use to compute $s'_k$.

\begin{lemma}
\label{lem:nesting-active-conv}
    For all $0 \le \ell \le \lm - 1$, and $([i_0', i_1'],k) \in S_{\ell}$, the unique level-$(\ell + 1)$ segment $([i_0, i_1],k)$  such that $[i_0', i_1'] \subseteq [i_0, i_1]$ is active, i.e., $([i_0, i_1],k) \in S_{\ell + 1}$.
\end{lemma}
\begin{proof}[Proof sketch]
    By a similar argument to that of \cref{lem:nesting-active}, we have $\Ah_{i}$ and $\Bh_{k-i}$ are constant on $[i_0,i_1]$. Since $\Ah_{i_0'}+\Bh_{k - i_0'}\ne \Ch_{k}$ and $i_0' \in [i_0, i_1]$, we conclude that $\Ah_{i_0}+\Bh_{k - i_0} \ne \Ch_{k}$, satisfying the first high-part condition of being active.

    For the congruence condition, we note that $[i_0', i_1'] \subseteq [i_0, i_1]$ and $([i_0, i_1],k)$ is a level-$(\ell + 1)$ segment.  By a similar argument to that of \cref{lem:nesting-active}, we have
    
    \[
    \bigl|(A_{i_0}+B_{k - i_0}-C_{k})-(A_{i_0'}+B_{k -i_0'}-C_{k})\bigr|< 2\cdot 2^{\ell + 1}.
    \]
    Since $([i_0', i_1'],k) \in S_{\ell}$, we have $A_{i_0'}+B_{k - i_0'}-C_{k}\equiv s'\pmod Q$ for some $s' \in [-4\cdot 2^\ell,\ 4\cdot 2^\ell]$. Let 
    \[
        s = s' + (A_{i_0}+B_{k - i_0}-C_{k})-(A_{i_0'}+B_{k - i_0'}-C_{k}).
    \]
    Then, $A_{i_0} + B_{k-i_0} - C_k \equiv s \pmod Q$ and $s \in [-4 \cdot 2^{\ell +1}, 4 \cdot 2^{\ell +1}]$.
\end{proof}

\begin{lemma}[Computing $s'_k$]
\label{lem:compute-sprime-conv}
    We can compute $s'_k$ for all $k\in[2,2n]$ in $\hO\!\left(n^{1+\mu}/Q\right)$ time.
\end{lemma}
\begin{proof}
    As in \cref{lem:compute-sprime}, we first compute $S_0$ and then compute the values $s'_k$ using $S_0$.
    
    \begin{enumerate}
        \item \textbf{Phase 1: Computing $S_0$.}
        We compute all level-$\lm$ segments via binary search and keep only the active ones to obtain $S_{\lm}$. Next, for $\ell=\lm-1,\ldots,0$, we construct $S_{\ell}$ from $S_{\ell+1}$ as follows. Each segment in $S_{\ell+1}$ refines into $\cO(1)$ level-$\ell$ subsegments, since decreasing $\ell$ by $1$ can split a maximal constant block into at most two blocks for $A$ and at most two blocks for $B$. We use binary search to find the boundary indices and compute these subsegments. For each resulting subsegment, we test whether it is active, and, if so, insert it into $S_{\ell}$.
    
        \item \textbf{Phase 2: Computing $s'_k$.}
        For each segment $([i_0,i_1],k)\in S_0$, we test whether $A_{i_0}+B_{k-i_0}\equiv C_k\pmod Q$. If the test succeeds, we increment $s'_k$ by $i_1-i_0+1$.
    \end{enumerate}
    
    The correctness of Phase~1 follows inductively using \cref{lem:nesting-active-conv}. For Phase~2, fix any $i \in [n]$ such that $k-i \in [n]$, $A_i+B_{k-i}\equiv C_k \pmod Q$, and $\Ah_i+\Bh_{k-i}\ne \Ch_k$. Let $([i_0,i_1],k)$ be the unique level-$0$ segment containing $(i,k)$. Since $([i_0,i_1],k)$ is a level-$0$ segment, the values $A_i$ and $B_{k-i}$ are constant for all $i\in[i_0,i_1]$. Because $i\in[i_0,i_1]$, we have
    \[
        \Ah_{i_0}+\Bh_{k-i_0}=\Ah_i+\Bh_{k-i}\ne \Ch_k
    \]
    and
    \[
        A_{i_0}+B_{k-i_0}\equiv A_i+B_{k-i}\equiv C_k \pmod Q.
    \]
    Hence, $([i_0,i_1],k)$ is an active level-$0$ segment and thus belongs to $S_0$.
    
    Moreover, the above identities hold for every $i\in[i_0,i_1]$. For a fixed $k$, the level-$0$ segments form disjoint maximal intervals, so exactly $i_1-i_0+1$ indices $i\in[n]$ in this segment satisfy these conditions. Therefore, it is correct that the algorithm increments $s'_k$ by exactly $(i_1-i_0+1)$.
    
    For the running time, \cref{fact:num-segments-conv} implies that the number of level-$\lm$ segments is $\cO(n^{1+\mu}/2^{\lm})$. Thus, computing $S_{\lm}$ takes $\tcO(n^{1+\mu}/2^{\lm})= \tcO(n^{1+\mu}/M)= \hO(n^{1+\mu}/Q)$ time. For the refinement step, by the guarantee on $Q$, for each $\ell$ we have $|S_\ell|= \hO(n^{1+\mu}/Q)$, and each segment in $S_{\ell+1}$ generates only $\cO(1)$ subsegments. Since there are $\cO(\log M)=\cO(\log n)$ levels, the total time for this refinement step is $\hO(n^{1+\mu}/Q)$. Finally, we scan all segments in $S_0$ and perform a constant-time test and update for each segment, so Phase~2 takes $\hO(|S_0|)= \hO(n^{1+\mu}/Q)$ time.
\end{proof}

\subsection{Finding a Good Modulus \texorpdfstring{$Q$}{Q}}
We now construct a \emph{good modulus} $Q$, that is, an integer $Q$ with
$M \le Q \le Mn^{o(1)}$ such that, for every level $0 \le \ell \le \lm$,
the number of active level-$\ell$ segments is at most $\hO(n^{1+\mu}/Q)$. The idea is essentially the same as in \cref{sec:find-Q}.

Let $4 \le R \le n^{o(1)}$ be a parameter to be fixed later, and let $\cP$ be the set of primes in $[R/2,R]$. We take $Q$ to be a product of primes $p_1,p_2,\ldots,p_T$, where $p_t \in \cP$ for all $1 \le t \le T$. We choose these primes inductively: at step $t \ge 1$, we select a prime $p_t \in \cP$ and multiply it into the current
modulus $Q_{t-1}$, obtaining
\[
Q_t := \prod_{i=1}^t p_i, \qquad \text{where } Q_0 := 1.
\]

For the analysis, it is convenient to track the following quantities:
\begin{enumerate}
    \item Let $X_{\ell,t}$ denote the number of pairs $(([i_0,i_1],k),s)$, where
    $([i_0,i_1],k)$ is a level-$\ell$ segment and $s\in[-4\cdot 2^\ell,\,4\cdot 2^\ell]$ such that
    \[
        A_{i_0}+B_{k-i_0}-C_k \equiv s \pmod{Q_t}
        \quad\text{and}\quad
        A_{i_0}+B_{k-i_0}-C_k \ne s.
    \]

    \item Let $Y_{\ell,t}$ denote the number of pairs $(([i_0,i_1],k),s)$, where
    $([i_0,i_1],k)$ is a level-$\ell$ segment and $s\in[-4\cdot 2^\ell,\,4\cdot 2^\ell]$, such that
    \[
        A_{i_0}+B_{k-i_0}-C_k \equiv s \pmod{Q_t}.
    \]
\end{enumerate}

Our inductive invariant is that
\begin{equation}
\label{eq:Xlt-conv}
    X_{\ell,t} = n^{1+\mu}\cdot \cO\!\left(\frac{\log^2 n}{R}\right)^t
    \qquad \text{for all } 0 \le \ell \le \lm.
\end{equation}
Later, we will choose $R$ so that $X_{\ell,T} = \hO(n^{1+\mu}/Q_T)$. By the following lemma, $Q_T$ is then a good modulus.

\begin{lemma}
\label{lem:bound-active-seg-conv}
    For every $0 \le \ell \le \lm$, we have $|S_{\ell}(Q_T)| \le X_{\ell,T}$.
\end{lemma}
\begin{proof}
    Let $([i_0,i_1],k)$ be an active level-$\ell$ segment with respect to $Q_T$. By \cref{def:active-segment-conv}, there exists $s\in[-4\cdot 2^\ell,\,4\cdot 2^\ell]$ such that $A_{i_0}+B_{k-i_0}-C_k \equiv s \pmod{Q_T}$ and $\Ah_{i_0}+\Bh_{k-i_0}-\Ch_k \ne 0$. We have $A_{i_0}+B_{k-i_0}-C_k \ne s$. Indeed $|A_{i_0}+B_{k-i_0}-C_k| \ge 7M/10$ by \cref{fact:big-diff-conv}, while $|s| \le 4\cdot 2^{\lm} < 4M/10$. Thus, each active segment contributes to $X_{\ell,T}$, and therefore $|S_{\ell}(Q_T)| \le X_{\ell,T}$.
\end{proof}

We now choose $p_t$ so that the invariant~\cref{eq:Xlt-conv} continues to hold. The base case $t=0$ holds for $Q_0:=1$ by \cref{fact:num-segments-conv}. Now assume~\cref{eq:Xlt-conv} holds for $t-1$ and that $Q_{t-1}$ has been fixed. For any prime $p\in\cP$, let $X_{\ell,t}(p)$ (resp., $Y_{\ell,t}(p)$) denote the value of $X_{\ell,t}$ (resp., $Y_{\ell,t}$) when $Q_t := Q_{t-1}\cdot p$. Our goal is to find a prime $p_t\in\cP$ such that~\cref{eq:Xlt-conv} holds at step $t$. Applying the following lemmas inductively yields the existence of the entire sequence $p_1,\ldots,p_T$.

\begin{lemma}
\label{lem:exists-pstar-conv}
    There exists $p^*\in\cP$ such that
    \begin{equation}
    \label{eq:p-star-conv}
        X_{\ell,t}(p^*) = X_{\ell,t-1}\cdot \cO\!\left(\frac{\log^2 n}{R}\right) \qquad\text{for all } 0\le \ell \le \lm.
    \end{equation}
    Moreover, for every $0\le \ell \le \lm$, if $p$ is chosen uniformly at random from $\cP$, then
    \[
        \Ex_{p\sim\cP}\!\bigl[X_{\ell,t}(p)\bigr] = X_{\ell,t-1}\cdot \cO\left(\frac{\log n}{R}\right).
    \]
\end{lemma}
\begin{proof}[Proof sketch]
    The proof is almost identical to that of \cref{lem:exists-pstar}. For completeness, we briefly indicate the changes.

    Fix a level $\ell$. Let $\cX_{\ell,t-1}$ be the set of pairs $(([i_0,i_1],k),s)$ counted by $X_{\ell,t-1}$, that is, those satisfying
    \[
        A_{i_0}+B_{k-i_0}-C_k \equiv s \pmod{Q_{t-1}}
        \qquad\text{and}\qquad
        A_{i_0}+B_{k-i_0}-C_k \ne s,
    \]
    where $([i_0,i_1],k)$ is a level-$\ell$ segment and $s\in[-4\cdot 2^\ell,\,4\cdot 2^\ell]$. For a prime $p\in\cP$, such a pair is counted in $X_{\ell,t}(p)$ if and only if, additionally, $p$ divides $\Delta := A_{i_0}+B_{k-i_0}-C_k-s$. Thus,
    \[
        X_{\ell,t}(p)
        =\sum_{(([i_0,i_1],k),s)\in \cX_{\ell,t-1}}
        \mathbf{1}\left[p \mid \bigl(A_{i_0}+B_{k-i_0}-C_k-s\bigr)\right].
    \]
    The remainder of the argument---bounding $\Pr_{p\sim\cP}[p\mid \Delta]$, taking expectations, and applying Markov's inequality with a union bound over $\ell\in\{0,\ldots,\lm\}$---is identical to that in \cref{lem:exists-pstar}, yielding a prime $p^*\in\cP$ that satisfies \cref{eq:p-star-conv} and the claimed expectation.
\end{proof}

\begin{lemma}[Computing $Y_{\ell,t}(p)$]
\label{lem:compute_Ylt-conv}
We can compute $Y_{\ell,t}(p)$ for all primes $p \in \cP$ and all levels $0 \le \ell \le \lm$ in $\tcO(R^2 Q_{t-1} n)$ time.
\end{lemma}

\begin{proof}
    Fix a prime $p\in\cP$ and a level $\ell$, and let $Q':=Q_{t-1}\cdot p$. We give an algorithm to compute $Y_{\ell,t}(p)$ for modulus $Q'$ in $\tcO(Q'n)$ time. Since $\lm=\cO(\log n)$, $|\cP|=\cO(R)$, and $Q'\le Q_{t-1}R$, the total time over all $p\in\cP$ and $\ell$ is 
    \[
         \tcO\bigl(|\cP|\cdot \lm \cdot Q' n\bigr) = \tcO\bigl(R^2 Q_{t-1} n\bigr).
    \]

    Recall that a level-$\ell$ segment $([i_0,i_1],k)$ is a maximal interval of valid indices $i$ on which both $\lfloor A_i/2^\ell\rfloor$ and $\lfloor B_{k-i}/2^\ell\rfloor$ are constant. Thus, for a fixed $k$, the segment start indices $i_0$ are exactly those valid indices $i\in[\max\{1,k-n\},\min\{n,k-1\}]$ such that either $i=\max\{1,k-n\}$, or $\lfloor A_{i-1}/2^\ell\rfloor\ne \lfloor A_i/2^\ell\rfloor$, or $\lfloor B_{k-i+1}/2^\ell\rfloor\ne \lfloor B_{k-i}/2^\ell\rfloor$. More concretely, define a boundary indicator array $I^A$ of size $n$ by
    \[
        I^A_1:=1,\qquad
        I^A_i:=\mathbf{1}\!\left[\left\lfloor A_{i-1}/2^\ell\right\rfloor
        \ne \left\lfloor A_i/2^\ell\right\rfloor\right]
        \quad (i\ge 2),
    \]
    and a boundary indicator array $I^B$ of size $n$ by
    \[
        I^B_n:=1,\qquad
        I^B_j:=\mathbf{1}\!\left[\left\lfloor B_{j+1}/2^\ell\right\rfloor
        \ne \left\lfloor B_j/2^\ell\right\rfloor\right]
        \quad (j\le n-1).
    \]
    Then, for a fixed $k$, the segment start indices $i_0$ are exactly those indices $i\in[\max\{1,k-n\},\min\{n,k-1\}]$ such that $I^A_i=1$ or $I^B_{k-i}=1$. We handle this OR condition by inclusion--exclusion.

    Work over the ring $\mathcal{R}:=\mathbb{F}[x]/(x^{Q'}-1)$, where $\mathbb{F}$ has characteristic greater than $n$. Define the following polynomials in $\mathcal{R}[y]$:
    \begin{align*}
        P_A^{\mathrm{all}}(x,y) &:= \sum_{i=1}^n x^{A_i\bmod Q'}\,y^i,
        &
        P_A^{\mathrm{bdry}}(x,y) &:= \sum_{i=1}^n I^A_i\,x^{A_i\bmod Q'}\,y^i,\\
        P_B^{\mathrm{all}}(x,y) &:= \sum_{j=1}^n x^{B_j\bmod Q'}\,y^j,
        &
        P_B^{\mathrm{bdry}}(x,y) &:= \sum_{j=1}^n I^B_j\,x^{B_j\bmod Q'}\,y^j.
    \end{align*}
    Compute the three products in $\mathcal{R}[y]$:
    \[
        P^{A}:=P_A^{\mathrm{bdry}}\cdot P_B^{\mathrm{all}},\qquad
        P^{B}:=P_A^{\mathrm{all}}\cdot P_B^{\mathrm{bdry}},\qquad
        P^{AB}:=P_A^{\mathrm{bdry}}\cdot P_B^{\mathrm{bdry}}.
    \]
    For any $k\in[2,2n]$ and $r\in\{0,\ldots,Q'-1\}$, let $U_k(r)$ be the coefficient of $x^r y^k$ in the polynomial $P^{A}+P^{B}-P^{AB}$. Then $U_k(r)$ equals the number of segment start indices $i_0\in[n]$ such that $A_{i_0}+B_{k-i_0}\equiv r\pmod{Q'}$ and either $I^A_{i_0}=1$ or $I^B_{k-i_0}=1$. Next define
    \[
        W(r):=\left|\left\{s\in[-4\cdot 2^\ell,\,4\cdot 2^\ell]:\ s\equiv r \pmod{Q'}\right\}\right| \qquad \text{for all } r\in\{0,\ldots,Q'-1\}.
    \]
    A level-$\ell$ segment $([i_0,i_1],k)$ contributes to $Y_{\ell,t}(p)$ if and only if
    \[
        A_{i_0}+B_{k-i_0}-C_k\equiv s \pmod{Q'}.
    \]
    Equivalently, if $r\equiv A_{i_0}+B_{k-i_0}\pmod{Q'}$, then $s\equiv r-C_k\pmod{Q'}$. Thus, grouping by $r$, the contribution of a fixed $k$ is
    \[
        \sum_{r=0}^{Q'-1} U_k(r)\cdot W\!\bigl((r-C_k)\bmod Q'\bigr),
    \]
    and therefore
    \[
        Y_{\ell,t}(p) =
        \sum_{k=2}^{2n}\ \sum_{r=0}^{Q'-1} U_k(r)\cdot W\!\bigl((r-C_k)\bmod Q'\bigr).
    \]

    For the running time, each polynomial product takes $\tcO(Q'n)$ time, and we compute only $\cO(1)$ such products. Computing $W$ costs $\cO(Q')$, and evaluating the double sum costs $\cO(Q'n)$. Hence, for fixed $p$ and $\ell$, the total time is $\tcO(Q'n)$.
\end{proof}

\begin{lemma}
\label{lem:finding-pt-conv}
    We can find a prime $p_t\in\cP$ satisfying \cref{eq:Xlt-conv}; that is $X_{\ell,t}(p_t) = n^{1 + \mu} \cdot \cO(\log^2 n/R)^t$ in $\tcO\bigl(R^2 Q_{t-1} n\bigr)$ time.
\end{lemma}

\begin{proof}[Proof sketch]
    We run the algorithm in \cref{lem:compute_Ylt-conv} to compute $Y_{\ell,t}(p)$ for all primes $p\in\cP$ and all levels $0\le \ell\le \lm$. For each $\ell$, define $Y^*_{\ell,t}:=\min_{p\in\cP} Y_{\ell,t}(p)$, and choose $p_t\in\cP$ that minimizes
    \[
        \Phi(p):=\max_{0\le \ell \le \lm}\bigl\{Y_{\ell,t}(p)-Y^*_{\ell,t}\bigr\}.
    \]
    The same argument as in \cref{lem:finding-pt} shows that this choice of $p_t$ ensures \cref{eq:Xlt-conv}. Moreover, the running time is dominated by computing all values $Y_{\ell,t}(p)$, which takes $\tcO\bigl(R^2 Q_{t-1} n\bigr)$ time by \cref{lem:compute_Ylt-conv}.
\end{proof}

It remains to set the parameters so that the above procedure yields a good modulus.

\begin{lemma}[Computing a Good Modulus]
\label{lem:compute-Q-conv}
We can compute a good modulus $Q$ in $\hO(Mn)$ time.
\end{lemma}
\begin{proof}
    By inductively applying \cref{lem:finding-pt-conv} for $T$ steps, we obtain primes $p_1,\ldots,p_T\in\cP$ and a modulus $Q_T:=\prod_{t=1}^T p_t$ such that
    \[
        X_{\ell,T}= n^{1+\mu}\cdot \cO\!\left(\frac{\log^2 n}{R}\right)^T \qquad\text{for all } 0\le \ell\le \lm.
    \]
    We take $T=\Theta(\log M/\log R)$ and $R=2^{\Theta\!\left(\sqrt{\log n/\log\log n}\right)}$. Then $M \le Q \le M n^{o(1)}$ and 
    \[
        X_{\ell,T}
        \le n^{1+\mu}\cdot \cO\!\left(\frac{\log^2 n}{R}\right)^T
        \le \frac{n^{1+\mu+o(1)}}{Q_T}
        \qquad\text{for all } 0\le \ell\le \lm,
    \]
    by the same argument as in \cref{lem:compute-Q}. Hence, $Q_T$ is a good modulus.
\end{proof}

\subsection{Combining the Results}

\begin{theorem}
\label{thm:exact-tri-mod-M-conv}
    There is a deterministic algorithm that solves \cref{prob:exact-tri-mod-M-conv} in time $\hO( Mn +n^{1+\mu}/M)$.
\end{theorem}

\begin{proof}
    We first compute a good modulus $Q$ via \cref{lem:compute-Q-conv} in time $\hO(Mn)$. We then compute $s_{k}$ for all $k \in [2, 2n]$ via \cref{lem:compute-s-conv} in time $ \tcO(Qn) = \hO(Mn)$. Finally, we compute $s'_{k}$ for all $k \in [2, 2n]$ via \cref{lem:compute-sprime-conv} in time $\hO\left(n^{1+\mu}/Q \right) = \hO\left(n^{1+\mu}/M \right)$. For each $k \in [2, 2n]$, we output \textsc{Yes} if and only if $s_{k} > s'_{k}$ and correctness follows from \cref{lem:s-vs-sprime-conv}. The total running time is $\hO( Mn +n^{1+\mu}/M)$.
\end{proof}

\begin{proof}[Proof of \cref{thm:main-conv}]
    By the reductions in \cref{lem:reduction-to-verification-conv}, it suffices to solve \cref{prob:exact-tri-mod-M-conv} efficiently. By \cref{thm:exact-tri-mod-M-conv}, there is a deterministic algorithm that solves \cref{prob:exact-tri-mod-M-conv} in time $\hO( Mn +n^{1+\mu}/M)$. Take $M=\Theta\!\left(n^{\mu/2}\right)$ by setting $d = \mu/2$. Then, the two terms balance, yielding the claimed time bound $\hO\!\left(n^{\mu/2 + 1}\right)$.
\end{proof}

\section*{Acknowledgments} 
The authors would like to thank Shyan Akmal and Adam Polak for pointing out issues in a previous version of this paper.

%% file: appendix.tex
\section{Derandomizing the +2 APSP algorithm  \texorpdfstring{of \cite{additiveapsp2026}}{}}

\label{appendix:apsp}

The only place where \cite{additiveapsp2026} used randomness is in the following \cref{lem:minplusgroups}.
Let $\MM(a,b,c)$ denote the time complexity of multiplying an $a\times b$ matrix with an $b\times c$ matrix. Using this notation, we have $\MM(n^{\gamma_1}, n^{\gamma_2},n^{\gamma_3}) = n^{\omega(\gamma_1,\gamma_2,\gamma_3)+o(1)}$.

\begin{lemma}[{\cite[Lemma 3.4]{additiveapsp2026}}]
   \label{lem:minplusgroups}
   Let integer parameter $L\ge 1$.
   Let $A\in \Z^{hd\times s}, B\in \Z^{s\times hd}$ be input matrices with entries in $[-U,U]$. 
   Partition the indices $i\in \{1,2,\dots,hd\}$ into $h$ contiguous groups each of size $d$ (so that $i$ belongs to the $\lceil i/d\rceil$-th group).

   Suppose $|A_{i,k}-A_{i',k}|\le L$ holds for all $k\in [s]$ and all $i,i'$ in the same group, and 
    $|B_{k,j}-B_{k,j'}|\le L$ holds for all $k\in [s]$ and all $j,j'$ in the same group.

Then, for any parameter $q\ge 1$, we can compute the Min-Plus product $A\star B$ by a randomized algorithm in time complexity
      \[\tcO\left (h^2 s + qL \cdot  \MM(hd,s,hd) + h^2 L\cdot  \MM(d,s/q,d)\right )\cdot \polylog(U). \]
   \end{lemma}

   Now we derandomize \cref{lem:minplusgroups}, thereby derandomizing the $+2$-approximate APSP algorithm of \cite{additiveapsp2026} in $O(n^{2.22548})$ time. For convenience, assume without loss of generality that $U\ge hds$, so that $\polylog(hds)$ can be suppressed by the $U^{o(1)}$ notation.
   \begin{lemma}
      \cref{lem:minplusgroups} can be derandomized with only $U^{o(1)}$-factor slowdown.
   \end{lemma}
   \begin{proof}[Proof]
   The only randomized part in the original proof of \cite{additiveapsp2026} is to pick a random prime modulus $p \in \Theta(q\log U)$. Below we describe how to derandomize this step using the standard recipe which finds a good modulus by iteratively deciding its small prime factors \cite{DBLP:conf/stoc/ChanL15}.
   We omit other parts of the original proof that are unrelated to derandomization.

    Following the original notations, use        the shorthand $\hat i = \lceil i/d\rceil \in [h]$ and $\hat j = \lceil  j/d\rceil\in [h]$ to denote the groups that contain indices $i\in [hd]$ and $j\in [hd]$ respectively. 
    Define matrix $A'\in \Z^{h\times s}$ by \[A'_{\hat i,k} \coloneqq \left\lfloor \tfrac{1}{L} \min_{(\hat i-1)d<i \le \hat id}  A_{i,k}\right\rfloor.\]
    Define the matrix $B' \in \Z^{s\times h}$ analogously.
    Let $C'\in \Z^{h\times h}$ be their Min-Plus  product $A'\star B'$, which can be computed by brute force in time $O(h^2s)$. 

    For a modulus $p$, \cite{additiveapsp2026} defined a triple $(\hat i,k,\hat j) \in [h]\times [s]\times [h]$ to be a \emph{false positive} (modulo $p$) if it does \emph{not} satisfy 
    \begin{equation}
   \left \lvert A'_{\hat i,k}+B'_{k,\hat j}-C'_{\hat i,\hat j}\right \rvert < 4, 
      \label{eqn:roundadditiveerr}
    \end{equation}
   but satisfies
   \begin{equation}
 A'_{ \hat i , k} + B'_{k, \hat j } - C'_{ \hat i, \hat j}\in   \{-3,\dots,3\} \pmod{p}. 
 \label{eqn:modpeq}
\end{equation}

We say a modulus $p$ is \emph{successful for} $(\hat i,\hat j)\in [h]\times [h]$ if 
\[|F_{\hat i,\hat j}^{(p)}|\coloneqq \Big \lvert \Big \{k\in [s]: (\hat i,k,\hat j)\text{ is a false positive modulo $p$}\Big \}\Big \rvert \le \frac{s}{q}.\]
One can compute $|F_{\hat i,\hat j}^{(p)}|$ for all $(\hat i,\hat j)\in [h]\times [h]$ in $O(h^2 s)$ total time by brute force.

\cite{additiveapsp2026} relied on the fact that,
for fixed $(\hat i,\hat j)\in [h]\times [h]$,
a random prime $p \in \Theta(q\log U)$ is successful with $\ge 0.99$ success probability. 
Therefore, \cite{additiveapsp2026} chose a random sequence of $O(\log h)$ many such moduli which can (with high probability)  make every pair $(\hat i,\hat j)\in [h]\times [h]$ successful at least once. 

We shall now deterministically choose $U^{o(1)}$ many moduli $m_1,m_2,\dots, m_{U^{o(1)}}\in O(q\cdot U^{o(1)})$ in $\tcO(h^2 s\cdot U^{o(1)})$ time to make every pair $(\hat i,\hat j)\in [h]\times [h]$  successful at least once. This would be sufficient for derandomizing the lemma (with $U^{o(1)}$-factor slowdown compared to the original time bound).

Initialize $P \coloneqq [h]\times [h]$, which denote the set of pairs for which no moduli have been successful yet. We repeat the following procedure $U^{o(1)}$ times, each time finding a modulus $m$ that makes at least $1/U^{o(1)}$ fraction of pairs in $P$ successful (and removing them from $P$): 

\begin{itemize}
    \item  Let $M = \exp({\sqrt{\log U}})$.  Initialize $Q_0 \coloneqq P$. Perform the following step for $I\coloneqq \lceil \log_M q\rceil $ iterations: In the $i$-th iteration,
    \begin{itemize}
        \item  Find a prime (by brute-force enumeration) $p_i \in [CM\log U,2CM\log U]$ (where $C>1$ is a sufficiently large absolute constant) such that  the set 
        \begin{equation}
        \label{eqn:qi}
         Q_{i}\coloneqq  \big \{(\hat i,\hat j)\in Q_{i-1} : |F_{\hat i,\hat j}^{(p_1p_2\cdots p_i)}| \le \frac{s}{M^{i}} \big \}
        \end{equation}
        has size
        \begin{equation}
        \label{eqn:qibound}
         | Q_{i}| \ge 0.9|Q_{i-1}|.
        \end{equation}
        Checking whether a prime $p_i$ is satisfactory can be done by brute force in in $O(h^2 s)$ time.
        The existence of such $p_i$ can be proved as follows: 
        By definition, we have 
        \begin{equation}
        F_{\hat i,\hat j}^{(p_1\cdots p_{i})} \subseteq \big\{ k \in F_{\hat i,\hat j}^{(p_1\cdots p_{i-1})} :  A'_{ \hat i , k} + B'_{k, \hat j } \in C'_{ \hat i, \hat j} + \{-3,\dots,3\} \pmod{p_i}.  \big\}.
        \label{eqn:fpi}
        \end{equation}
        For fixed $(\hat i,\hat j)\in Q_{i-1}$ and $k\in F_{\hat i,\hat j}^{(p_1\cdots p_{i-1})}$, we know  \cref{eqn:roundadditiveerr} must be violated, so $\{A'_{\hat i,k}+B'_{k,\hat j} - C'_{\hat i,\hat j} - r\}_{r\in \{-3,\dots,3\}}$ are all non-zero integers bounded by $O(U)$, and thus together have at most $O(\log_M U)$ prime factors larger than $M$. Since the number of primes in the interval $[CM\log U,2CM\log U]$ is at least $\Omega((M\log U)/\log (M\log U))$ by the prime number theorem,  the probability that a random prime $p_i$ from this interval divides any of $\{A'_{\hat i,k}+B'_{k,\hat j} - C'_{\hat i,\hat j} - r\}_{r\in \{-3,\dots,3\}}$ is at most $\frac{1}{10M}$ (for large enough constant $C$). Consequently, by \cref{eqn:fpi} and linearity of expectation, 
        \[\Ex_{p_i}[|F_{\hat i,\hat j}^{(p_1\cdots p_{i})}|] \le \tfrac{1}{10M}|F_{\hat i,\hat j}^{(p_1\cdots p_{i-1})}|.\]
        By Markov's inequality,
        \[\Pr_{p_i}\Big[|F_{\hat i,\hat j}^{(p_1\cdots p_{i})}|\le \tfrac{1}{M}|F_{\hat i,\hat j}^{(p_1\cdots p_{i-1})}|\Big ] \ge 0.9.\]
        Recall that $|F_{\hat i,\hat j}^{(p_1\cdots p_{i-1})}| \le \frac{s}{M^{i-1}}$ for every $(\hat i,\hat j)\in Q_{i-1}$. Hence, by  \cref{eqn:qi} and linearity of expectation, we have
        \[\Ex_{p_i}[|Q_i|]\ge 0.9 |Q_{i-1}|.\]
        In particular, there exists a prime $p_i$ such that $|Q_i|\ge 0.9|Q_{i-1}|$ as claimed.
        \end{itemize}
        After $I=\lceil \log M_q\rceil$ iterations, we have found primes $p_1,\dots,p_I$. Define modulus $m = p_1p_2\cdots p_I$. By \cref{eqn:qibound} (recalling $Q_0=P$) we have 
        \[\big \lvert\big\{(\hat i,\hat j)\in P: |F^{(m)}_{\hat i,\hat j}|  \le \frac{s}{M^I}\big \}\big \rvert \ge |Q_{I}|\ge |P|\cdot 0.9^{I}.  \]
        Note that $0.9^I = 0.9^{\lceil \log_M q\rceil } \ge 0.9\cdot 0.9^{\log q/\log M} = 0.9\cdot 0.9^{\log q/\sqrt{\log U}} \ge 0.9\cdot 0.9^{\sqrt{\log U}} \ge 1/U^{o(1)}$.
        Note that $M^I = M^{\lceil \log_{M} q\rceil} \ge q$.  So
        \[\big \lvert\big\{(\hat i,\hat j)\in P: |F^{(m)}_{\hat i,\hat j}|  \le \frac{s}{q}\big \}\big \rvert \ge |Q_{I}|\ge |P|/U^{o(1)},  \]
        that is, the modulus $m$ makes at least $1/U^{o(1)}$ fraction of the pairs in $P$ successful, as wanted.
        Note that the modulus $m \le (2CM\log U)^I \le  (2CM\log U)^{1 + \log_{M} q} = qM \cdot (2C\log U)^{1 + \log_{M} q} \le qM \cdot (2C\log U)^{1+\sqrt{\log U}} \le q \cdot U^{o(1)}$ as required.
\end{itemize}
The total time complexity of deterministically finding the moduli is obviously $O(h^2 s\cdot U^{o(1)})$.
\end{proof}